\documentclass[a4paper,11pt]{article}

\usepackage[labelfont=bf,labelsep=period]{caption}
\usepackage[nodayofweek]{datetime}
\usepackage{enumitem}
\usepackage{float}
\usepackage[margin=2.5cm]{geometry}
\usepackage{graphicx}
\usepackage{hyperref}
\usepackage{numbertabbing}
\usepackage{times}
\usepackage{url}
\usepackage{xcolor}
\usepackage{xspace}

\usepackage{amsmath} 
\allowdisplaybreaks[2]          
\usepackage{amssymb} 
\usepackage{amsthm} 
\newtheoremstyle{plain-boldhead}
  {\topsep}
  {\topsep}
  {\itshape}
  {}
  {\bfseries}
  {.}
  { }
  {\thmname{#1}\thmnumber{ #2}\thmnote{ (\bfseries #3)}}
\newtheoremstyle{definition-boldhead}
  {\topsep}
  {\topsep}
  {\normalfont}
  {}
  {\bfseries}
  {.}
  { }
  {\thmname{#1}\thmnumber{ #2}\thmnote{ (\bfseries #3)}}
\theoremstyle{plain-boldhead}
\newtheorem{theorem}{Theorem}

\newtheorem{lemma}[theorem]{Lemma}
\newtheorem{corollary}[theorem]{Corollary}

\theoremstyle{definition-boldhead}
\newtheorem{definition}{Definition}

\newtheorem{example}{Example}

\newdateformat{simple}{\THEDAY\ \monthname[\THEMONTH]\ \THEYEAR}
\simple

\floatstyle{ruled}
\newfloat{algo}{htbp}{algo}
\floatname{algo}{Algorithm}

\providecommand{\naive}{na\"{i}ve\xspace}
\providecommand{\Naive}{Na\"{i}ve\xspace}

\def \ifempty#1{\def\temp{#1} \ifx\temp\empty }
\renewcommand{\P}{\mathrm{P}}

\newcommand{\str}[1]{\textsc{#1}}
\newcommand{\var}[1]{\textit{#1}}
\newcommand{\op}[1]{\textsl{#1}}
\newcommand{\msg}[2]{\ensuremath{\ifempty{#2} [\str{#1}] \else [\str{#1}, {#2}] \fi}}

\newcommand{\nil}{\ensuremath{\bot}}
\newcommand{\false}{\textsc{false}\xspace}
\newcommand{\true}{\textsc{true}\xspace}

\newcommand{\BC}{\ensuremath{\mathbb{C}}\xspace}

\newcommand{\BF}{\ensuremath{\mathbb{F}}\xspace}

\newcommand{\BK}{\ensuremath{\mathbb{K}}\xspace}

\newcommand{\BQ}{\ensuremath{\mathbb{Q}}\xspace}

\newcommand{\CA}{\ensuremath{\mathcal{A}}\xspace}
\newcommand{\CB}{\ensuremath{\mathcal{B}}\xspace}
\newcommand{\CC}{\ensuremath{\mathcal{C}}\xspace}

\newcommand{\CF}{\ensuremath{\mathcal{F}}\xspace}
\newcommand{\CG}{\ensuremath{\mathcal{G}}\xspace}
\newcommand{\CH}{\ensuremath{\mathcal{H}}\xspace}

\newcommand{\CK}{\ensuremath{\mathcal{K}}\xspace}

\newcommand{\CP}{\ensuremath{\mathcal{P}}\xspace}
\newcommand{\CQ}{\ensuremath{\mathcal{Q}}\xspace}

\newcommand{\CS}{\ensuremath{\mathcal{S}}\xspace}
\newcommand{\CT}{\ensuremath{\mathcal{T}}\xspace}

\newcommand{\CX}{\ensuremath{\mathcal{X}}\xspace}

\newcommand\ts{\var{ts}\xspace}

\newcommand\wts{\var{wts}\xspace}
\newcommand\rid{\var{rid}\xspace}

\begin{document}

\title{\bf Asymmetric Distributed Trust\protect{$^1$}}

\author{Orestis Alpos$^{2,4}$\\
  University of Bern \& Common Prefix\\
  \url{oralpos@gmail.com}
  \and Christian Cachin$^2$\\
  University of Bern\\
  \url{christian.cachin@unibe.ch}
  \and Bj\"{o}rn Tackmann$^3$\\
  DFINITY\\
  \url{bjoern@dfinity.org}
  \and Luca Zanolini$^4$\\
  Ethereum Foundation\\
  \url{luca.zanolini@ethereum.org}
}

\footnotetext[1]{This work combines multiple preliminary publications on
  asymmetric trust and protocols with asymmetric trust, which appeared at
  OPODIS~2019~\protect\cite{DBLP:conf/opodis/CachinT19},
  DISC~2021~\protect\cite{DBLP:conf/wdag/CachinZ21}, and
  ESORICS-CBT~2021~\protect\cite{DBLP:conf/esorics/CachinZ21}}

\footnotetext[2]{Institute of Computer Science, University of Bern,
  Neubr\"{u}ckstrasse 10, 3012 CH-Bern, Switzerland.}

\footnotetext[3]{DFINITY, CH-8000 Z\"{u}rich, Switzerland.  Work done at IBM
  Research~-- Zurich.}

\footnotetext[4]{Work done at the University of Bern.}

\date{\today}

\maketitle

\begin{abstract}\noindent
  Quorum systems are a key abstraction in distributed fault-tolerant
  computing for capturing trust assumptions.  They can be found at the core
  of many algorithms for implementing reliable broadcasts, shared memory,
  consensus and other problems.  This paper introduces \emph{asymmetric
    Byzantine quorum systems} that model subjective trust.  Every process
  is free to choose which combinations of other processes it trusts and
  which ones it considers faulty.  Asymmetric quorum systems strictly
  generalize standard Byzantine quorum systems, which have only one global
  trust assumption for all processes.  This work also presents protocols
  that implement abstractions of shared memory, broadcast primitives, and a
  consensus protocol among processes prone to Byzantine faults and
  asymmetric trust.  The model and protocols pave the way for realizing
  more elaborate algorithms with asymmetric trust.
\end{abstract}

\section{Introduction}
\label{sec:intro}

Byzantine quorum systems~\cite{DBLP:journals/dc/MalkhiR98} are a fundamental primitive for
building resilient distributed systems from untrusted components.  Given a
set of nodes, a quorum system captures a trust assumption on the nodes in
terms of potentially malicious protocol participants and colluding groups
of nodes.  Based on quorum systems, many well-known algorithms for
\emph{reliable broadcast}, \emph{shared memory}, \emph{consensus} and more
have been implemented; these are the main abstractions to synchronize the
correct nodes with each other and to achieve consistency despite the
actions of the faulty, so-called \emph{Byzantine} nodes.

Traditionally, trust in a Byzantine quorum system for a set of
processes~\CP has been \emph{symmetric}.  In other words, a global
assumption specifies which processes may fail, such as the simple and
prominent \emph{threshold quorum} assumption, in which any subset of \CP of
a given maximum size may collude and act against the protocol.  The most
basic threshold Byzantine quorum system, for example, allows all subsets of
up to $f < n/3$ processes to fail.  Some classic works also model
arbitrary, non-threshold symmetric quorum
systems~\cite{DBLP:journals/dc/MalkhiR98,DBLP:journals/joc/HirtM00},
but it is unknown if these have been used in practice.

However, trust is inherently subjective.
\emph{De gustibus non est disputandum~-- There is no disputing about taste.}
Estimating which processes will function correctly and which ones will
misbehave may depend on personal taste.  A myriad of local choices
influences one process' trust in others, especially because there are so
many forms of ``malicious'' behavior.  Some processes might not even be
aware of all others, yet a process should not depend on unknown third
parties in a distributed collaboration.  How can one model asymmetric trust
in distributed protocols?  Can traditional Byzantine quorum systems be
extended to subjective failure assumptions?  How do the standard protocols
generalize to this model?

\paragraph{Asymmetric trust.}
In this paper, we answer these questions and introduce models and protocols
for asymmetric distributed trust.  We formalize \emph{asymmetric
  (Byzantine) quorum systems} for asynchronous protocols, in which every
process can make its own assumptions about Byzantine faults of others.  We
introduce several protocols with asymmetric trust that strictly generalize
the existing algorithms, which require common trust.

Our formalization takes up earlier work by Damg{\aa}rd et al.~\cite{DBLP:conf/asiacrypt/DamgardDFN07}
and starts out with the notion of a fail-prone system that forms the basis
of a symmetric Byzantine quorum system.  A global fail-prone system for a
process set \CP contains all maximal subsets of \CP that might jointly fail
during an execution.  In an asymmetric quorum system, every process
specifies its \emph{own} fail-prone system and a corresponding set of local
quorums.  These local quorum systems satisfy a \emph{consistency condition}
that ranges across all processes and a local \emph{availability condition},
and generalize symmetric Byzantine quorum system according to Malkhi and
Reiter~\cite{DBLP:journals/dc/MalkhiR98}.

\paragraph{Protocols with asymmetric quorums.}
Quorum systems are used within various fault-tolerant distributed
protocols, here specifically within protocols for systems subject to
Byzantine faults.  An important aspect of our notion concerns its
relation to existing protocols: it should be easy to generalize the known
protocols to the asymmetric model, ideally simply by replacing the symmetric
quorums with their asymmetric counterparts.
Indeed this is the case for many, but not for all
protocols described here.  A different, generalized analysis is necessary
in any case.

We show first that two existing protocols for emulating a shared regular
register also work in the asymmetric model.  Second, we introduce
asymmetric Byzantine consistent and reliable broadcast primitives, for
which we again only change the quorums compared to the protocols with
symmetric quorums.  Third, we address consensus, one of the most important
primitives in distributed computing, and extend a randomized binary
consensus protocol for asynchronous networks to work with asymmetric trust.
The protocol relies on a common coin abstraction, for which a different
implementation is needed.

Our randomized consensus takes up the award-winning, randomized, and
signature-free implementation of consensus by Most{\'{e}}faoui \emph{et
  al.}~\cite{DBLP:conf/podc/MostefaouiMR14}.  In its 2014 version, however,
this protocol suffered from a liveness issue, which was corrected
subsequently~\cite{DBLP:journals/jacm/MostefaouiMR15}, although the fix
added considerable complexity.  The corrected algorithm offers the same
asymptotic complexity in message and time as the original algorithm, but it
requires more communication steps.

Through our randomized asymmetric consensus, we also introduce a novel way
of fixing the problem in the original protocol.  It retains the latter
protocol's simplicity, which is an appealing property.  Obviously, our
asymmetric consensus protocol can also be instantiated with symmetric
threshold quorums to work in the same model as the protocol of
Most{\'{e}}faoui \emph{et al.}~\cite{DBLP:conf/podc/MostefaouiMR14}.  In
order to clearly demonstrate the liveness issue and to show how our
approach avoids it, we also include in this work a discussion of this
randomized consensus algorithm in the symmetric-trust model.

In the traditional models for quorum-based systems, all correct processes
uniformly benefit from the guarantees of a protocol as long as the initial
assumption expressed by the fail-prone system holds.  With subjective trust,
this symmetry no longer exists.  Some of the correct processes may have made
assumptions that proved appropriate in an execution with actually faulty
processes $F \subset \CP$; we call these processes \emph{wise}.  Other correct
processes, however, may have assumed that only a proper subset of $F$ actually
fails; these processes are \emph{\naive} and they do not enjoy the same
guarantees as the wise ones, even though they are correct.  In particular, our
protocols typically ensure safety only for wise processes and liveness depends
on the existence of a sufficiently large group of wise processes.

\paragraph{Motivation.}
Interest in consensus protocols based on Byzantine quorum systems has
surged recently because of their application to permissioned blockchain
networks~\cite{DBLP:conf/wdag/CachinV17,DBLP:conf/eurosys/AndroulakiBBCCC18}.
Typically run by a consortium, such distributed ledgers often use
\emph{Byzantine-fault tolerant (BFT)} protocols like
PBFT~\cite{DBLP:journals/tocs/CastroL02},
Tendermint~\cite{DBLP:journals/corr/abs-1807-04938}, or
HotStuff~\cite{DBLP:conf/podc/YinMRGA19} for consensus that rely on
symmetric threshold quorum systems.  The Bitcoin blockchain and many other
cryptocurrencies, which triggered this development, started from different
assumptions and use so-called permissionless protocols, in which everyone
may participate.  Those algorithms capture the relative influence of the
participants on consensus decisions by an external factor, such as invested
``work'' or ``stake'' in the system.

A middle ground between permissionless blockchains and BFT-based ones has
been introduced by the blockchain networks of Ripple
(\url{https://ripple.com}) and Stellar (\url{https://stellar.org}).  Their
stated model for achieving network-level consensus uses subjective trust in
the sense that each process declares a local list of processes that it
``trusts'' in the protocol.

Consensus in the \emph{Ripple} blockchain (and for the \emph{XRP}
cryptocurrency on the \emph{XRP Ledger}) is executed by its validator
nodes.  Each node declares a \emph{Unique Node List (UNL)}, which are
validators that this node trusts, in the sense that ``the given participant
believes [they] will not conspire to defraud [the node].''
%
%
At least up to around 2020, however, nodes have not really been
free in their trust choice since ``Ripple
provides a default and recommended list which we [Ripple] expand based on watching
the history of validators operated by Ripple and third parties''
\cite{ripple-faq20}.  
%
%
%
As of 2023, the XRP ledger documentation states that ``currently the XRP
Ledger Foundation and Ripple are known to publish recommended default lists
of high quality validators \dots''~\cite{ripple-faq23}.
It is clear that two nodes that transact via the XRP ledger need to have
some validators that they trust in common.  But many questions are left
open about the kind of decentralization offered by the Ripple protocol.

\emph{Stellar} was created as an evolution of Ripple that shares much of
the same design philosophy.  The Stellar consensus protocol~\cite{Mazieres2015TheSC}
powers the \emph{Stellar Lumen (XLM)} cryptocurrency and introduces
\emph{federated Byzantine quorum systems (FBQS)}; they also capture 
subjective trust assumptions, but differ technically from asymmetric
quorum systems.
Stellar's consensus protocol uses \emph{quorum slices}, which are ``the
subset of a quorum that can convince one particular node of agreement.''
In an FBQS, ``each node chooses its own quorum slices'' and ``the
system-wide quorums result from these decisions by individual
nodes''~\cite{stellar15}.

\paragraph{Contribution.}
The main motivation for this work is to understand how existing ideas of
subjective trust, as manifested in the Ripple and Stellar blockchains,
relate to traditional quorum systems.  The formalization of asymmetric
quorums provides a sound foundation for protocols with asymmetric trust.
The protocols described here generalize well-known, classic algorithms in
the literature and therefore look similar.  This should be seen as a
feature, actually, because simplicity and modularity are important guiding
principles in science.

Our contributions are as follows:
\begin{itemize}
\item We introduce asymmetric Byzantine quorum systems formally in
  Section~\ref{sec:quorums} as an extension of standard Byzantine quorum
  systems and discuss some of their properties.
\item In Section~\ref{sec:memory}, we show two implementations of a shared
  register, with single-writer, multi-reader regular semantics, using
  asymmetric Byzantine quorum systems.
\item We examine broadcast primitives in the Byzantine model with
  asymmetric trust in Section~\ref{sec:broadcast}.  In particular, we
  define and implement Byzantine consistent and reliable broadcast
  protocols.
\item In Section~\ref{sec:consensus}, we present the first asynchronous
  Byzantine consensus protocol with asymmetric trust.  It uses
  randomization, provided by an asymmetric common coin protocol, to
  circumvent the impossibility of asynchronous consensus.
\end{itemize}
Before presenting the technical contributions, we discuss related work in
Section~\ref{sec:related} and state our system model in
Section~\ref{sec:model}.  A detailed discussion of the liveness issue in
the existing signature-free Byzantine consensus
protocol~\cite{DBLP:conf/podc/MostefaouiMR14} and of our approach to fixing
it appears in Appendix~\ref{app:attack}.

\section{Related work}
\label{sec:related}

\paragraph{Practical systems: Ripple and Stellar.}
The \emph{Ripple} consensus protocol is run by an open set of validator
nodes.  The protocol uses votes, similar to standard consensus protocols,
whereby each validator only communicates with the validators in its UNL.
Each validator chooses its own UNL, which makes it possible for anyone to
participate, in principle, similar to proof-of-work blockchains.  Early
investigations suggested that the intersection of the UNLs of every two
validators should be at least 20\% of each list~\cite{scyobr14}, assuming
that also less than one fifth of the validators in the UNL of every node
might be faulty.  An independent analysis by Armknecht et
al.~\cite{DBLP:conf/trust/ArmknechtKMYZ15} later argued that this bound
must be more than~40\%.  A technical report of Chase and
MacBrough~\cite[Thm.~8]{DBLP:journals/corr/abs-1802-07242} concludes, under
the same assumption of $f < n/5$ faulty nodes in every UNL of size~$n$,
that the UNL overlap should actually be at least~90\%.

However, the same paper also derives a counterexample to the liveness of
the Ripple consensus
protocol~\cite[Sec.~4.2]{DBLP:journals/corr/abs-1802-07242} as soon as two
validators don't have ``99\% UNL overlap.''  By generalizing the example,
this essentially means that the protocol can get stuck \emph{unless all
  nodes have the same UNL}.  According to the standards of the field of
distributed systems, though, a protocol needs to satisfy safety \emph{and}
liveness because achieving only one of these properties is trivial.
Amores-Sesar \emph{et al.}~\cite{DBLP:conf/opodis/Amores-SesarCM20} confirm
the prior analysis and exhibit a wider set of examples how safety and
liveness may be violated in executions of the Ripple consensus protocol.
They first show that the network may fork, even under the standard
condition stated by Ripple on the overlap of UNLs, and then that the
consensus protocol may lose liveness in the presence of only one Byzantine
process, even if all the processes have the same UNL.  These works,
however, exploit arbitrary message delays, i.e., a period of asynchronous
network behavior, which is not assumed by Ripple and arguably also unlikely
to occur in practice.

The \emph{Stellar consensus protocol (SCP)} also features open membership
and lets every node express its own set of trusted
nodes~\cite{Mazieres2015TheSC,DBLP:conf/sosp/LokhavaLMHBGJMM19}.
Generalizing from Ripple's flat lists of unique nodes, every node declares
a collection of trusted sets called \emph{quorum slices}, whereby a slice
is ``the subset of a quorum convincing one particular node of agreement.''
A \emph{quorum} in Stellar is a set of nodes ``sufficient to reach
agreement,'' defined as a set of nodes that contains one slice for each
member node.  The quorum choices of all nodes together yield a
\emph{federated Byzantine quorum systems (FBQS)}.  The literature on
Stellar gives properties for FBQS and contains protocols that build on
them, which have been implemented in the Stellar
blockchain~\cite{DBLP:conf/sosp/LokhavaLMHBGJMM19}.  However, standard
Byzantine quorum systems and FBQS are \emph{not} comparable because (1) an
FBQS when instantiated with the same trust assumption for all processes
does not reduce to a symmetric quorum system and (2) existing protocols do
not directly generalize to FBQS.

\paragraph{Models of asymmetric trust.}
Starting from Stellar's notions, Garc{\'{\i}}a{-}P{\'{e}}rez and
Gotsman~\cite{DBLP:conf/opodis/Garcia-PerezG18} build a link from FBQS to
existing quorum-system concepts by investigating a Byzantine reliable
broadcast abstraction in an FBQS.  They show that the \emph{federated
  voting protocol} of Stellar~\cite{Mazieres2015TheSC} is similar to
Bracha's reliable broadcast~\cite{DBLP:journals/iandc/Bracha87} and that it
implements a variation of Byzantine reliable broadcast on an FBQS for
executions that contain, additionally, a set of so-called intact nodes.
Losa \emph{et al.}~\cite{DBLP:conf/wdag/LosaGM19} have later formulated an
abstraction of the consensus mechanism in the Stellar network by
introducing \emph{Personal Byzantine quorum systems} (PBQS).  In contrast
to the other notions of ``quorums'', their definition does not require a
global intersection among quorums. This may lead to several separate
\emph{consensus clusters} such that each one satisfies agreement and
liveness on its own.

The FBQS and PBQS concepts, however, differ from the notion of a Byzantine
quorum system in the literature.  In particular, the characterization of
their properties seems to take into account knowledge of which nodes are
Byzantine, and their effects are therefore analyzed in the context of
particular executions.  Existing notions of symmetric quorum systems in the
literature~\cite{DBLP:journals/dc/MalkhiR98,DBLP:journals/joc/HirtM00}
start from an a-priori assumption about all potentially faulty sets of
nodes, through a fail-prone system~\cite{DBLP:journals/dc/MalkhiR98}.  This
permits to study protocol-independent aspects of quorum systems.

Another approach for designing Byzantine fault-tolerant (BFT) consensus
protocols has been introduced by Malkhi \emph{et
  al.}~\cite{DBLP:conf/ccs/MalkhiN019}, namely \emph{Flexible BFT}. This
notion guarantees higher resilience by introducing a new
\emph{alive-but-corrupt} fault type, which denotes processes that attack
safety but not liveness.  Malkhi \emph{et
  al.}~\cite{DBLP:conf/ccs/MalkhiN019} also define \emph{flexible Byzantine
  quorums} that allow processes in the system to have different faults
models.

Our work, in contrast, goes back to the model of Damg{\aa}rd et
al.~\cite{DBLP:conf/asiacrypt/DamgardDFN07}.  It already contains the basic
formulation of asymmetric trust and expresses it in the context of
synchronous protocols for secure distributed computation with
process-specific fail-prone systems.  The model features only a consistency
property, but omits liveness.  Damg{\aa}rd et
al.~\cite{DBLP:conf/asiacrypt/DamgardDFN07} also state a characterization
of when an asymmetric Byzantine quorum system exists (with the so-called
$B^3$), but give no proof.  Their work has remained without impact until
research on cryptocurrencies has revived interest in heterogeneous and
subjective trust models.

\paragraph{Signature-free randomized consensus.}
Most{\'{e}}faoui \emph{et al.}~\cite{DBLP:conf/podc/MostefaouiMR14} present
a randomized, signature-free, and round-based asynchronous consensus
algorithm for binary values. It achieves optimal resilience and takes
$O(n^2)$ constant-sized messages. Randomization is achieved through a
common coin as defined by Rabin~\cite{DBLP:conf/focs/Rabin83}.  Their
binary consensus algorithm has been taken up for constructing the ``Honey
Badger BFT'' protocol by Miller \emph{et
  al.}~\cite{DBLP:conf/ccs/MillerXCSS16}, for instance.  One important
contribution of Most{\'{e}}faoui \emph{et
  al.}~\cite{DBLP:conf/podc/MostefaouiMR14} is a new binary validated
broadcast primitive with a non-deterministic termination property; it has
also found applications in other protocols~\cite{DBLP:conf/nca/CrainGLR18}.

Tholoniat and Gramoli~\cite{TG19} observe a liveness issue in the protocol
by Most{\'{e}}faoui \emph{et al.}~\cite{DBLP:conf/podc/MostefaouiMR14} in
which an adversary is able to prevent progress among the correct processes
by controlling messages between them and by sending them values in a
specific order.

In a later work, Most{\'{e}}faoui \emph{et
  al.}~\cite{DBLP:journals/jacm/MostefaouiMR15} present a different version
of their randomized consensus algorithm that does not suffer from the
liveness problem anymore. The resulting algorithm offers the same
asymptotic complexity in message and time as their previous
algorithm~\cite{DBLP:conf/podc/MostefaouiMR14}, but requires more
communication steps.

\section{System model}
\label{sec:model}

\paragraph{Processes.}
We consider a system of $n$ \emph{processes} $\CP = \{p_1, \dots, p_n\}$
that communicate with each other.  The processes interact asynchronously
with each other through exchanging messages. The system itself is
asynchronous, i.e., the delivery of messages among processes may be delayed
arbitrarily and the processes have no synchronized clocks.  Every process
is identified by a name, but such identifiers are not made explicit.
A protocol for \CP consists of a collection of programs with instructions
for all processes.  Protocols are presented in a modular way using the
event-based notation of Cachin et al.~\cite{DBLP:books/daglib/0025983}.

\paragraph{Executions and faults.}
An \emph{execution} starts with all processes in a special initial state;
subsequently the processes repeatedly trigger events, react to events, and
change their state through computation steps.
Every execution is \emph{fair} in the sense that, informally, processes do
not halt prematurely when there are still steps to be taken or events to be
delivered (we refer to the standard literature for a formal
definition~\cite{Lynch96}).

A process that follows its protocol during an execution is called
\emph{correct}.  On the other hand, a \emph{faulty} process may crash or
even deviate arbitrarily from its specification, e.g., when
\emph{corrupted} by an adversary; such processes are also called
\emph{Byzantine}.  We consider only Byzantine faults here and assume for
simplicity that the faulty processes fail right at the start of an
execution.

\paragraph{Functionalities.}
A \emph{functionality} is an abstraction of a distributed computation,
either a primitive that may be used by the processes or a service that they
will provide.  Every functionality in the system is specified through its
\emph{interface}, containing the events that it exposes to protocol
implementations that may call it, and its \emph{properties}, which define
its behavior.  A process may react to a received event by changing their
state and triggering further events.

There are two kinds of events in an interface: \emph{input events} that the
functionality receives from other abstractions, typically to invoke its
services, and \emph{output events}, through which the functionality
delivers information or signals a condition to a process.  The behavior of a
functionality is usually stated through a number of properties or through a
sequential implementation.

Multiple functionalities may be composed together modularly.  In a modular
protocol implementation, in particular, every process executes the program
instructions of the protocol implementations for all functionalities in
which it participates.

\paragraph{Links.} 
We assume there is a low-level functionality for sending messages over
point-to-point links between each pair of processes.  In a protocol, this
functionality is accessed through the events of ``sending a message'' and
``receiving a message.''  Point-to-point messages are authenticated and
delivered reliably among correct processes.

Moreover, we assume FIFO ordering on the reliable point-to-point links for
every pair of correct processes.  This means that if a correct process has
``sent'' a message $m_1$ and subsequently ``sent'' a message $m_2$, then
every correct process does not ``receive'' $m_2$ unless it has earlier also
``received'' $m_1$.  FIFO-ordered links are actually a very common
assumption. Protocols that guarantee FIFO order on top of (unordered)
reliable point-to-point links are well-known and simple to
implement~\cite{10.5555/HadzilacosT93,DBLP:books/daglib/0025983}.  We
remark that there is only one FIFO-ordered reliable point-to-point link
functionality in the model; hence, FIFO order holds among the messages
exchanged by the implementations for \emph{all} functionalities used by a
protocol.

\paragraph{Idealized digital signatures.}

A \emph{digital signature scheme} provides two operations, $\op{sign}_i$
and $\op{verify}_i$.  The invocation of $\op{sign}_i$ specifies a
process~$p_i$ and takes a bit string~$m \in \{0,1\}^*$ as input and returns
a signature $\sigma \in \{0,1\}^*$ with the response.  Only $p_i$ may
invoke $\op{sign}_i$.  The operation $\op{verify}_i$ takes a putative
signature~$\sigma$ and a bit string~$m$ as parameters and returns a Boolean
value with the response.  Its implementation satisfies that
$\op{verify}_i(\sigma ,m)$ returns \true for any $i \in [1,n]$ and
$m \in \{0,1\}^*$ if and only if $p_i$ has executed $\op{sign}_i(m)$ and
obtained $\sigma$ before; otherwise, $\op{verify}_i(\sigma ,m)$ returns
\false.  Every process may invoke \op{verify}.


\section{Asymmetric Byzantine quorum systems}
\label{sec:quorums}

This section defines asymmetric Byzantine quorum systems and the notions of
a guild and a tolerated system, which are used in protocols later.  To set
the stage, symmetric Byzantine quorum systems are reviewed first.

\subsection{Review of symmetric trust}
\label{sec:symquorums}

Quorum systems are well-known in settings with symmetric trust.  As
demonstrated by many applications to distributed systems, ordinary quorum
systems~\cite{DBLP:journals/siamcomp/NaorW98} and Byzantine quorum systems~\cite{DBLP:journals/dc/MalkhiR98} play
a crucial role in formulating resilient protocols that tolerate faults
through replication~\cite{CharronBostPS10}.  A quorum system typically
ensures a consistency property among the processes in an execution, despite
the presence of some faulty processes.

For the model with Byzantine faults, \emph{Byzantine quorum systems} have
been introduced by Malkhi and Reiter~\cite{DBLP:journals/dc/MalkhiR98}.  This notion is
defined with respect to a \emph{fail-prone system} $\CF \subseteq 2^{\CP}$,
a collection of subsets of \CP, none of which is contained in another, such
that some $F \in \CF$ with $F \subseteq \CP$ is called a \emph{fail-prone
  set} and contains all processes that may at most fail together in some
execution~\cite{DBLP:journals/dc/MalkhiR98}.  A fail-prone system is the same as the
\emph{basis} of an \emph{adversary structure}, which was introduced
independently by Hirt and Maurer~\cite{DBLP:journals/joc/HirtM00}.

A fail-prone system captures an assumption on the possible failure patterns
that may occur.  It specifies all maximal sets of faulty processes that a
protocol should tolerate in an execution; this means that a protocol
designed for \CF achieves its properties as long as the set $F$ of actually
faulty processes satisfies $F \in \CF^*$.  Here and from now on, the
notation $\CA^*$ for a system $\CA \subseteq 2^\CP$, denotes the collection
of all subsets of the sets in $\CA$, that is,
\(\CA^* = \{ A' | A' \subseteq A, A \in \CA \}\).

\begin{definition}[Byzantine quorum system~\cite{DBLP:journals/dc/MalkhiR98}]\label{def:quorum}
  A \emph{Byzantine quorum system} for \CF is a collection of sets of
  processes $\CQ \subseteq 2^{\CP}$ where no set is contained in
  another and each $Q \in \CQ$ is called a
  \emph{quorum}, such the following properties hold:
  \begin{description}
  \item[Consistency:] The intersection of any two quorums contains at least
    one process that is not faulty, i.e.,
    \[\forall Q_1, Q_2 \in \CQ , \forall F \in \CF: \, Q_1 \cap Q_2 \not
    \subseteq F.\]
\item[Availability:] For any set of processes that may fail together, there
  exists a disjoint quorum in \CQ, i.e.,
  \[\forall F \in \CF: \, \exists Q \in \CQ: \, F \cap Q = \emptyset.\]
  \end{description}
\end{definition}
The above notion is also known as a \emph{Byzantine dissemination quorum
  system}~\cite{DBLP:journals/dc/MalkhiR98} and allows a protocol to be
designed despite arbitrary behavior of the potentially faulty processes.
The notion generalizes the usual threshold failure assumption for Byzantine
faults~\cite{DBLP:journals/jacm/PeaseSL80}, which considers that any set of
$f$ processes may fail.

We say that a set system \CT \emph{dominates} another set system \CS if for
each $S \in \CS$ there is some $T \in \CT$ such that
$S \subseteq T$~\cite{DBLP:journals/jacm/Garcia-MolinaB85}.  In this sense,
a quorum system for \CF is \emph{minimal} whenever it does not dominate any
other quorum system for~\CF.  A \emph{maximal} set system is defined
analogously.

Similarly to the threshold case, where $n > 3f$ processes are needed to
tolerate $f$ faulty ones in many Byzantine protocols, Byzantine quorum
systems can only exist if not ``too many'' processes fail.

\begin{definition}[$Q^3$-condition~\cite{DBLP:journals/dc/MalkhiR98,DBLP:journals/joc/HirtM00}]
  A fail-prone system \CF satisfies the \emph{$Q^3$-condition}, abbreviated
  as $Q^3(\CF)$, whenever it holds
  \[
    \forall F_1, F_2, F_3 \in \CF: \, \CP \not\subseteq F_1 \cup F_2 \cup F_3.
  \]
\end{definition}
In other words, $Q^3(\CF)$ means that no \emph{three} fail-prone sets
together cover the whole system of processes.  A $Q^k$-condition can be
defined like this for any $k \geq 2$~\cite{DBLP:journals/joc/HirtM00}.

The following result of Malkhi and
Reiter~\cite[Theorem~5.4]{DBLP:journals/dc/MalkhiR98} considers the
\emph{bijective complement} of a process set $\CS \subseteq 2^{\CP}$, which
is defined as $\overline{\CS} = \{ \CP \setminus S | S \in \CS \}$, and
turns \CF into a Byzantine quorum system.  A related theorem was formulated
also by Hirt and Maurer~\cite{DBLP:journals/joc/HirtM00}.

\begin{lemma}\label{lem:canon}
  Given a fail-prone system \CF, a Byzantine quorum system for \CF exists
  if and only if~$Q^3(\CF)$.
  In particular, if $Q^3(\CF)$ holds, then $\overline{\CF}$, the bijective
  complement of \CF, is a Byzantine quorum system.
\end{lemma}
The quorum system $\CQ = \overline{\CF}$ is called the \emph{canonical
  quorum system} of~\CF.  According to the duality between \CQ and \CF,
properties of \CF are sometimes ascribed to \CQ as well.
However, note that the canonical quorum system is not always minimal.  For
instance, if \CF consists of all sets of $f \ll n/3$ processes, then each
quorum in the canonical quorum system has $n-f$ members, but also the
family of all subsets of \CP with $\lceil \frac{n+f+1}{2} \rceil < n-f$
processes forms a quorum system.

\paragraph{Core sets.}

A \emph{core set}~$C$ for \CF is a minimal set of processes that contains
at least one correct process in every execution.  More precisely,
$C \subseteq \CP$ is a core set whenever (1) for all $F \in \CF$, it holds
$\CP \setminus F \cap C \neq \emptyset$ (and, equivalently,
$C \not\subseteq F$) and (2) for all $C' \subsetneq C$, there exists
$F \in \CF$ such that $\CP \setminus F \cap C' = \emptyset$ (and,
equivalently, \(C' \subseteq F\)).  With the threshold failure assumption,
every set of $f+1$ processes is a core set.  A \emph{core-set system} \CC
is the minimal collection of all core sets, in the sense that no set in \CC
is contained in another.

Core sets can be complemented by \emph{survivor sets}, as shown by
Junqueira et al.~\cite{DBLP:journals/dc/JunqueiraMHP10}.  This yields a dual characterization of
resilient distributed protocols, which parallels ours using fail-prone sets
and quorums.

\paragraph{Kernels.}

Given a symmetric Byzantine quorum system~\CQ, we define a
\emph{kernel}~$K$ as a minimal set of processes that overlaps with every
quorum. A kernel generalizes the notion of a \emph{core
  set}~\cite{DBLP:conf/icdcs/JunqueiraM03}.
 
\begin{definition}[Kernel system]
  A set $K \subseteq \CP$ is a \emph{kernel} of a quorum system~\CQ if an only if
  \[
    \forall Q \in \CQ: \, K \cap Q \neq \emptyset
  \]
  and 
    \[
    \forall K' \subsetneq K: \, \exists~Q \in \mathcal{Q}: \, Q \cap K' = \emptyset.
  \]
  
  We also define the \emph{kernel system} $\CK$ of \CQ to be the set of all kernels of $\CQ$.  
\end{definition}

For example, under a threshold failure assumption where any $f$ processes
may fail, every set of $\big\lfloor\frac{n-f+1}{2}\big\rfloor$ processes is
a kernel. In particular, $n=3f+1$ if and only if every kernel has $f+1$
processes.

The definition of a kernel is related to that of a core set in the
following sense. 

\begin{lemma}\label{lem:kernelcore}
  Let \CF be a fail-prone system and \(\CQ = \overline{\CF}\) be the
  canonical quorum system of \CF.  Then the kernel system of \CQ is the
  same as the core-set system for~\CF.
\end{lemma}

\begin{proof}
Consider a kernel system \( \mathcal{K} \) of a Byzantine quorum system \( \mathcal{Q} \). By definition, the following two properties hold with respect to every kernel $K \in \CK$: 
\begin{itemize}
\item[(i)] For every quorum \( Q \) in \( \mathcal{Q} \), the intersection with the kernel \( K \) is non-empty, i.e., \( K \cap Q \neq \emptyset \).
\item[(ii)] For any proper subset \( K' \) of \( K \), there exists a quorum \( Q \) in \( \mathcal{Q} \) such that \( K' \) does not intersect with \( Q \), i.e., \( Q \cap K' = \emptyset \).
\end{itemize}
Given the canonical quorum system \( \mathcal{Q} \) derived from the fail-prone system \( \mathcal{F} \), by definition of canonical quorum system of \CF we have that for every \( Q\) in \( \mathcal{Q} \), there exists a unique fail-prone set \( F\) in \( \mathcal{F}\) such that \( Q\) is precisely the complement of \( F\) within \( \mathcal{P} \), that is, \( Q = \mathcal{P} \setminus F \). Consequently, the concepts of a kernel and a core set are equivalent in this context, as a core set is defined with respect to sets of the form \( \mathcal{P} \setminus F \).
\end{proof}

\begin{lemma}\label{lem:kernelinquorum}
  Let \CF, \CQ, and \CK be a fail-prone system, a Byzantine quorum system
  for \CF, and the kernel system of \CQ, respectively. Then, for every
  quorum $Q \in \CQ$, there exists a kernel $K \in \CK$ such that
  $K \subseteq Q$.
\end{lemma}
 
\begin{proof}
  Consider the quorum system \CQ for \CF. Let \( F \) be any such
  fail-prone set in \CF. For a given quorum \( Q \in \CQ \), define the set
  \( K = Q \setminus F \). By definition, \( K \) is a subset of \( Q \),
  i.e., \( K \subseteq Q \).  The consistency property of the Byzantine
  quorum system now implies that any two quorums \( Q, Q' \in \CQ \) have
  an intersection \( Q \cap Q' \) that is not fully contained within
  \( F \). Therefore, \( K \) intersects with \( Q' \) since
  \( (Q \setminus F) \cap Q' = K \cap Q' \) is not empty.  This property
  holds for every \( Q' \in \CQ \) and confirms that \( K \) intersects
  with every quorum in \( \mathcal{Q} \). As such, \( K \) satisfies the
  first property of a kernel of \CQ.

  For the second property, minimality, let us consider such a \( K \). To
  construct a kernel contained in $Q$, we progressively remove elements
  from \( K \), ensuring that the resultant subset retains the property of
  intersection with all quorums. This process terminates with a subset
  \( K^* \), which cannot be reduced further without losing the
  intersection property. The minimality of \( K^* \) is guaranteed by the
  contradiction that arises from the assumption that a proper subset of
  \( K^* \) could intersect with all quorums, as this would violate the
  termination of our removal process. Therefore, \( K^* \) is a kernel by
  definition since it is the minimal intersecting set with every quorum in
  \( \mathcal{Q} \), and it is contained within the original quorum \( Q \)
  from which we subtracted \( F \).  This shows that \( K^* \) is a kernel
  of~\( Q \).
\end{proof}

\subsection{Asymmetric trust}
\label{sec:asymquorums}

In our model with asymmetric trust, every process is free to make its own
trust assumption and to express this with a fail-prone system.  Hence, an
\emph{asymmetric fail-prone system} $\BF = [\CF_1, \dots, \CF_n]$ consists
of an array of fail-prone systems, where $\CF_i$ denotes the trust
assumption of~$p_i$.  One often assumes $p_i \not\in F_i$ for practical
reasons, but this is not necessary.  This notion has earlier been
formalized by Damg{\aa}rd et al.~\cite{DBLP:conf/asiacrypt/DamgardDFN07}.

\begin{definition}[Asymmetric Byzantine quorum system]\label{def:asymquorum}
  An \emph{asymmetric Byzantine quorum system} for \BF is an array of
  collections of sets $\BQ = [\CQ_1, \dots, \CQ_n]$, where
  $\CQ_i \subseteq 2^{\CP}$ for $i \in [1,n]$.  The set
  $\CQ_i \subseteq 2^{\CP}$ is called the \emph{quorum system of $p_i$} and
  any set $Q_i \in \CQ_i$ is called a \emph{quorum (set) for $p_i$}.  It
  satisfies:
  \begin{description}
  \item[Consistency:] The intersection of two quorums for any two processes
    contains at least one process for which either process assumes that it
    is not faulty, i.e.,
    \[
      \forall i,j \in [1,n],
      \forall Q_i \in \CQ_i, \forall Q_j \in \CQ_j,
      \forall F_{ij} \in {\CF_i}^* \cap {\CF_j}^*: \,
      Q_i \cap Q_j \not\subseteq F_{ij}.
    \]
  \item[Availability:] For any process~$p_i$ and any set of processes that
    may fail together according to $p_i$, there exists a disjoint quorum
    for $p_i$ in $\CQ_i$, i.e.,
    \[
      \forall i \in [1,n],
      \forall F_i \in \CF_i: \, \exists Q_i \in \CQ_i: \, F_i \cap Q_i =
      \emptyset.
    \]
\end{description}
\end{definition}

Recall that the consistency condition for a (symmetric) Byzantine quorum
system requires that at least one process in the intersection of every two
quorums is correct.  In the asymmetric case, quorums are subjective and
defined according to the quorum system for each process.  The asymmetric
consistency property states that in the intersection of every two
subjective quorums of two processes there exists at least one process that
is correct according to one of the two processes.  On the other hand, the
availability condition in the above definition is a direct extension of the
symmetric case, since it considers the quorum system of each process
separately.  We remark that availability suffices for implementing some
protocols but a stronger assumption (i.e., the existence of a guild,
introduced below) is needed for others.

The existence of asymmetric quorum systems can be characterized with a
property that generalizes the $Q^3$-condition for the underlying asymmetric
fail-prone systems as follows.

\begin{definition}[$B^3$-condition]
\label{def:b3}
  An asymmetric fail-prone system \BF satisfies the
  \emph{$B^3$-condition}, abbreviated as $B^3(\BF)$, whenever it holds that
  \[
    \forall i,j \in [1,n],
    \forall F_i \in \CF_i, \forall F_j\in\CF_j,
    \forall F_{ij} \in {\CF_i}^*\cap{\CF_j}^*: \,
    \CP \not\subseteq F_i \cup F_j \cup F_{ij} 
  \]
\end{definition}

The following result is the generalization of Lemma~\ref{lem:canon} for
asymmetric quorum systems; it was stated by Damg{\aa}rd et
al.~\cite{DBLP:conf/asiacrypt/DamgardDFN07} without proof.

\begin{theorem}\label{thm:asymcanon}
  An asymmetric fail-prone system \BF satisfies $B^3(\BF)$ if and only if
  there exists an asymmetric quorum system for \BF.
\end{theorem}

\begin{proof}
  Suppose that $B^3(\BF)$.  We let $\BQ = [\CQ_1, \dots, \CQ_n]$, where
  $\CQ_i = \overline{\CF_i}$ is the canonical quorum system of $\CF_i$,
  and show that \BQ is an asymmetric quorum system.  Indeed, let
  \(Q_i \in \CQ_i\), \(Q_j \in \CQ_j\), and
  \(F_{ij} \in {\CF_i}^* \cap {\CF_j}^*\) for any $i$ and~$j$. Then
  \(F_i = \CP \setminus Q_i \in \CF_i\) and
  \(F_j = \CP \setminus Q_j \in \CF_j\) by construction, and therefore,
  \(F_i \cup F_j \cup F_{ij} \not= \CP\) holds according to $B^3(\BF)$.
  This means there is some
  \(p_k \in \CP \setminus (F_i \cup F_j \cup F_{ij})\).  Because
  $p_k \not \in F_i$, it holds $p_k \in Q_i$ and analogously $p_k \in Q_j$.
  This implies in turn that \(p_k \in Q_i \cap Q_j\) but
  \(p_k \notin F_{ij}\) and proves the consistency condition.  The
  availability property holds by construction of the canonical quorum
  systems.

  To show the reverse direction, let \BQ be a candidate asymmetric
  Byzantine quorum system for \BF that satisfies availability and assume
  towards a contradiction that \(B^3(\BF)\) does not hold.  We show that
  consistency cannot be fulfilled for \BQ. By our assumption there are sets
  \(F_i, F_j, F_{ij}\) in \BF such that \(F_i \cup F_j \cup F_{ij} = \CP\),
  which means also that \(\CP \setminus (F_i \cup F_j) \subseteq
  F_{ij}\). The availability condition for \BQ then implies that there are
  sets \(Q_i \in \CQ_i\) and \(Q_j \in \CQ_j\) with
  \(F_i \cap Q_i = \emptyset\) and \(F_j \cap Q_j = \emptyset\). Now for
  every \(p_k \in Q_i \cap Q_j\) it holds that \(p_k \notin F_i \cup F_j\) by
  availability and therefore \(p_k \in \CP \setminus (F_i \cup F_j)\).  Taken
  together this means that
  \(Q_i \cap Q_j \subseteq \CP \setminus (F_i \cup F_j) \subseteq F_{ij}\).
  Hence, \(\BQ\) does not satisfy the consistency condition and the
  statement follows.
\end{proof}

\paragraph{Asymmetric core sets and kernels.}

Let \(\BF = [ \CF_1, \dots, \CF_n ]\) be an asymmetric fail-prone system.
An \emph{asymmetric core-set system}~\BC is an array of
collections of sets \([ \CC_1, \dots, \CC_n ]\) such that each \(\CC_i\) is
a core set system for the fail-prone system \(\CF_i\).  We call a set
$C_i \in \CC_i$ a \emph{core set for $p_i$}.

Given an asymmetric quorum system \BQ for \BF, an \emph{asymmetric kernel
  system} for \BQ is defined analogously as the array
$\BK = [\CK_1, \dots, \CK_n]$ that consists of the kernel systems for all
processes in~\CP.  A set $K_i \in \CK_i$ is called a \emph{kernel
  for~$p_i$.}  This means that every kernel for $p_i$ has a non-empty
intersection with every quorum of~$p_i$.

\paragraph{\Naive and wise processes.}

Recall that the guarantees of quorum-based protocols apply to \emph{correct}
processes only, but not to faulty ones.
The faults or corruptions occurring in a protocol execution with an underlying
quorum system induce a set~$F$ of actually \emph{faulty processes}.  However,
no process knows~$F$ and this information is only available to an observer
outside the system.  With a traditional quorum system \CQ designed for a
fail-prone set \CF, the guarantees of a protocol usually hold as long as
$F \in \CF^*$, and if $F$ is not contained in $\CF^*$, no useful properties
can be derived for any process.

With asymmetric quorums, we further distinguish between two kinds of
correct processes, depending on whether they considered $F$ in their trust
assumption or not.  Given a protocol execution, the processes are therefore
partitioned into three types:
\begin{description}
\item[Faulty:] A process $p_i \in F$ is \emph{faulty}.
\item[\Naive:] A correct process $p_i$ for which $F \not\in {\CF_i}^*$
  is called \emph{\naive}.
\item[Wise:] A correct process $p_i$ for which $F \in {\CF_i}^*$ is called
  \emph{wise}.
\end{description}

The \naive processes are new for the asymmetric case, as all correct processes
are wise under a symmetric trust assumption.  Protocols for asymmetric quorums
cannot guarantee the same properties for \naive processes as for wise ones,
since the \naive processes may have the ``wrong friends.''
In one formalization of the Stellar protocol, correct nodes that find
themselves in a similar situation have been called
``befouled''~\cite{Mazieres2015TheSC}.

\begin{example}
\label{ex:wisenaive}
We define an example of asymmetric fail-prone system $\BF_A$ on
$\CP = \{p_1, p_2, p_3, p_4, p_5\}$.  The notation $\Theta^n_k(\CS)$ for a
set \CS with $n$ elements denotes the ``threshold'' combination operator
and enumerates all subsets of \CS of cardinality~$k$.  W.l.o.g.\ every
process trusts itself.  The diagram below shows fail-prone sets as shaded
areas and the notation $\mbox{}_{k}^{n}$ in front of a fail-prone set
stands for $k$ out of the $n$ processes in the set.

\vspace*{-2ex}
\begin{minipage}[c]{0.05\linewidth}
  \vspace*{2ex}
  \Large$\BF_A$:
\end{minipage}
\begin{minipage}[c]{0.4\linewidth}
  \begin{align*}
    \mbox{}\\
    \CF_1 & = \Theta^4_1(\{p_2,p_3,p_4,p_5\}) \\
    \CF_2 & = \Theta^4_1(\{p_1,p_3,p_4,p_5\}) \\
    \CF_3 & = \Theta^2_1(\{p_1,p_2\}) \ast \Theta^2_1(\{p_4,p_5\}) \\
    \CF_4 & = \Theta^4_1(\{p_1,p_2,p_3,p_5\}) \\
    \CF_5 & = \{\{p_2,p_4\}\} \\
  \end{align*}
\end{minipage}
\begin{minipage}[c]{0.5\linewidth}
  \includegraphics[scale=0.42]{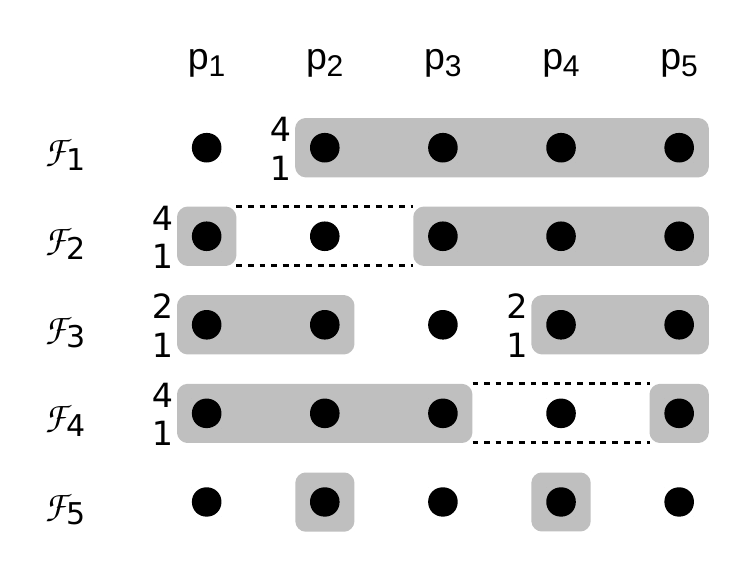}
\end{minipage}
The operator $\ast$ for two sets satisfies
$\CA \ast \CB = \{ A \cup B | A \in \CA, B \in \CB \}$.

As one can verify in a straightforward way, $B^3(\BF_A)$ holds.  Let
$\BQ_A$ be the canonical asymmetric quorum system for $\BF_A$.  Note that
since $\BF_A$ contains the fail-prone systems of $p_3$ and $p_5$ that
permit two faulty processes each, this fail-prone system cannot be obtained
as a special case of $\Theta^5_1(\{p_1, p_2, p_3, p_4, p_5\})$.  When
$F = \{p_2, p_4\}$, for example, then processes~$p_3$ and $p_5$ are wise
and $p_1$ is~\naive.
\end{example}

\paragraph{Guilds.}

If too many processes are \naive or even fail during a protocol run with
asymmetric quorums, then protocol properties cannot be ensured. A
\emph{guild} is a set of wise processes that contains at least one quorum
for each member; by definition this quorum consists only of wise processes.
A guild ensures liveness and consistency for typical protocols.  This
generalizes from protocols with symmetric trust, where the correct
processes in every execution form a quorum by definition.  A guild
represents a group of influential and well-connected wise processes, like
in the real world.

\begin{definition}[Guild]
  Given a fail-prone system \BF, an asymmetric quorum system \BQ for \BF,
  and a protocol execution with faulty processes~$F$, a \emph{guild \CG for
    $F$ and \BQ} satisfies two properties:
  \begin{description}
  \item[Wisdom:] \CG is a set of wise processes:
    \[
      \forall p_i \in \CG :\, F \in {\CF_i}^*.
    \]
  \item[Closure:] \CG contains a quorum for each of its members:
    \[
      \forall p_i \in \CG :\, \exists Q_i \in \CQ_i :\, Q_i \subseteq \CG.
    \]
  \end{description}
\end{definition}

A guild is related to an ``intact set'' in the Stellar consensus
protocol~\cite{Mazieres2015TheSC,DBLP:conf/sosp/LokhavaLMHBGJMM19}, but the
two notions differ in how they are defined.
Observe that the union of two guilds is again a guild, since the union
consists only of wise processes and contains again a quorum for each
member.  All guilds overlap, as the next result shows.

\begin{lemma}\label{lem:guildunique}
  In any execution with a guild $\CG$, every two guilds intersect.
\end{lemma}

\begin{proof}
  Let $\CP$ be a set of processes, $\CG$ be a guild, and $F$ be the set of
  actually faulty processes. Furthermore, suppose that there is another
  guild $\CG'$.  Let $p_i \in \CG$ and $p_j \in \CG'$ be two processes and
  consider a quorum $Q_i \subseteq \CG$ for $p_i$ and a quorum
  $Q_j \subseteq \CG'$ for $p_j$. From the definition of an asymmetric
  quorum system it must hold $Q_i \cap Q_j \nsubseteq F$, with
  $Q_i \cap Q_j \neq \emptyset$ and $F \in {\CF_i}^* \cap {\CF_j}^*$. It
  follows that there exists a wise process $p_k \in Q_i \cap Q_j$ with
  $p_k \in \CG$ and $p_k \in \CG'$. Notice also that $\CG$ and $\CG'$ both
  contain a quorum for $p_k$.
\end{proof}

It follows that every execution with a guild contains a unique
\emph{maximal guild}~$\CG_{\max}$.  The next lemma shows that if a guild
exists, no quorum for any process contains only faulty processes.

\begin{lemma}\label{lem:quorumonlybyz}
  Let $\CG_{\text{max}}$ be the maximal guild for a given execution and let $\mathbb{Q}$ be the canonical asymmetric quorum system. Then, there cannot be a quorum $Q_j \in \mathcal{Q}_j$ for any process $p_j$ consisting only of faulty processes.
\end{lemma}

\begin{proof}
  Given an execution with $F$ as set of faulty processes, suppose there is
  a guild $\CG_{\text{max}}$. This means that for every process
  $p_i \in \CG_{\text{max}}$, a quorum $Q_i \subseteq \CG_{\text{max}}$
  exists such that $Q_i \cap F = \emptyset$. It follows that for every
  $p_i \in \CG_{\text{max}}$, there is a set $F_i \in \mathcal{F}_i$ such
  that $F \subseteq F_i$. Recall that since \BQ is a quorum system,
  $B^3(\BF)$ holds. From Definition~\ref{def:b3}, we have that for all
  $i,j \in [1,n]$, all $F_i \in \CF_i, \forall F_j\in\CF_j$, and all
  $F_{ij} \in {\CF_i}^*\cap{\CF_j}^*$, it holds
  $\CP \not\subseteq F_i \cup F_j \cup F_{ij}$.
  
  Towards a contradiction, assume that there is a process $p_j$ such that
  there exists a quorum $Q_j \in \mathcal{Q}_j$ for $p_j$ with $Q_j =
  F$. This implies that there exists $F_j \in \mathcal{F}_j$ such that
  $F_j = \mathcal{P} \setminus F$.
  
  Let $F_i$ be the fail-prone system of $p_i \in \CG_{\text{max}}$ such
  that $F \subseteq F_i$ and let $F_j = \mathcal{P} \setminus F$ as just
  defined. Then, $F_i \cup F_j \cup F_{ij} = \mathcal{P}$. This follows
  from the fact that $F_i$ contains $F$ and that
  $F_j = \mathcal{P} \setminus F$.  This contradicts the $B^3$-condition
  for \BF.
\end{proof}

\begin{lemma}\label{lem:quorumnaive}
  Let $\CG_{\text{max}}$ be the maximal guild for a given execution and let $p_i$ be any correct process. Then, every quorum for $p_i$ contains at
  least one process in $\CG_{\text{max}}$.
\end{lemma}

\begin{proof}
  The claim naturally derives from the consistency property of an
  asymmetric quorum system.  Consider any correct process $p_i$ and one of
  its quorums, $Q_i \in \CQ_i$.  For any process
  $p_j \in \CG_{\text{max}}$, let $Q_j$ be a quorum of $p_j$ such that
  $Q_j \subseteq \CG_{\text{max}}$, which exists because $\CG_{\text{max}}$
  is a guild.  Then, the quorum consistency property implies that
  $Q_i \cap Q_j \neq \emptyset$.  Thus, $Q_i$ contains a process
  in the maximal guild.
\end{proof}

Finally, we show with an example that it is possible for a wise process to
be outside the maximal guild.

\begin{example}
\label{ex:wisegmax}
Let us consider a seven-process asymmetric quorum system $\BQ_B$, defined through its fail-prone system~$\BF_B$.

\vspace*{-1ex}
\noindent
\begin{minipage}[c]{0.05\linewidth}
  \vspace*{2ex}
  \center\Large$\BF_B$:
\end{minipage}
\begin{minipage}[l]{0.45\linewidth}
  \begin{eqnarray*}
    \mbox{}\\
    \CF_1 & = & \Theta^3_2(\{p_2,p_4,p_5\}) \ast \{p_6\} \ast \{p_7\}\\
    \CF_2 & = & \Theta^3_2(\{p_3,p_4,p_5\}) \ast \{p_6\} \ast \{p_7\} \\
    \CF_3 & = & \Theta^3_2(\{p_1,p_4,p_5\}) \ast \{p_6\} \ast \{p_7\}\\
    \CF_4 & = & \Theta^4_1(\{p_1,p_2,p_3,p_5\}) \ast \{p_6\} \ast \{p_7\}\\
    \CF_5 & = & \Theta^4_1(\{p_1,p_2,p_3,p_4\}) \ast \{p_6\} \ast \{p_7\}\\
    \CF_6 & = & \Theta^3_3(\{p_1,p_3,p_7\})\\
    \CF_7 & = & \Theta^3_3(\{p_3,p_4,p_5\})\\
    \mbox{}
  \end{eqnarray*}
\end{minipage}
\begin{minipage}[c]{0.5\linewidth}
  \includegraphics[scale=0.42]{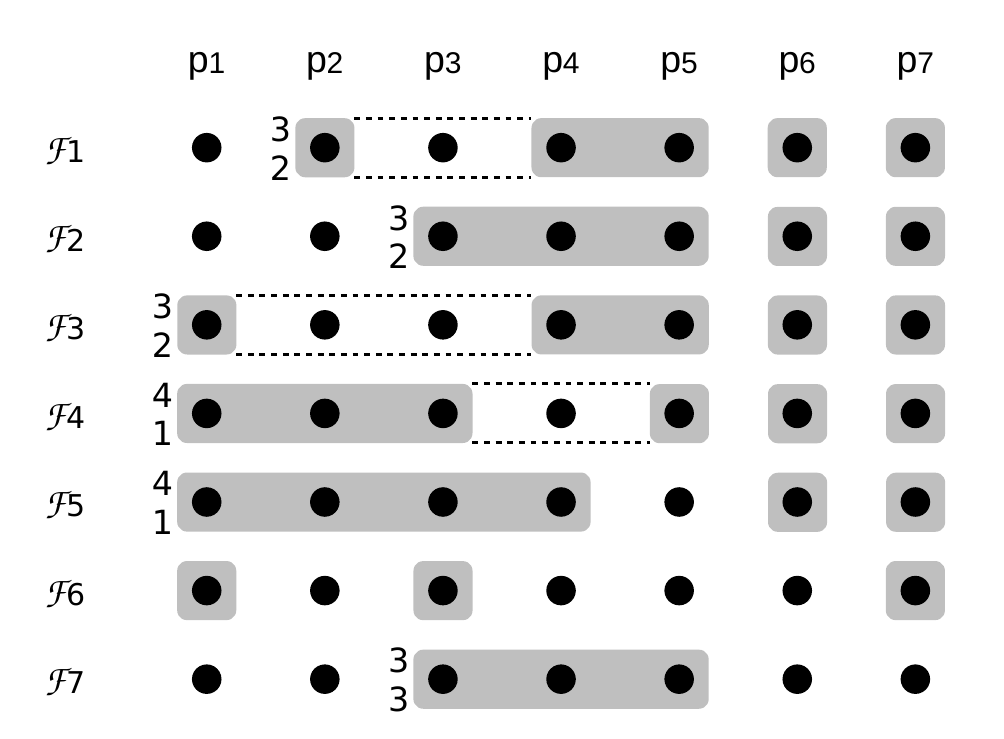}
\end{minipage}
\end{example}

\noindent
One can verify that $B^3(\BF_B)$ holds; hence, let $\BQ_B$ be the
canonical quorum system of $\BF_B$.

\noindent
\begin{minipage}[c]{0.05\linewidth}
  \vspace*{2ex}
  \center\Large$\BQ_B$:
\end{minipage}
\begin{minipage}[l]{0.8\linewidth}
  \begin{eqnarray*}
   \CQ_1 &=& \{\{p_1, p_3, p_5 \}, \{p_1, p_3, p_4 \}, \{p_1, p_2, p_3 \}\} \\
   \CQ_2 &=& \{\{p_1, p_2, p_5 \}, \{p_1, p_2, p_4 \}, \{p_1, p_2, p_3 \}\} \\
   \CQ_3 &=& \{\{p_2, p_3, p_5 \}, \{p_2, p_3, p_4 \}, \{p_1, p_2, p_3 \}\} \\
   \CQ_4 &=& \{\{p_1, p_2, p_3, p_4 \}, \{p_1, p_2, p_4, p_5 \},
             \{p_1, p_3, p_4, p_5 \}, \{p_2, p_3, p_4, p_5 \}\} \\
   \CQ_5 &=& \{\{p_1, p_2, p_3, p_5 \}, \{p_1, p_2, p_4, p_5 \},
             \{p_1, p_3, p_4, p_5 \}, \{p_2, p_3, p_4, p_5 \}\} \\
   \CQ_6 &=& \{\{p_2, p_4, p_5, p_6\}\} \\
   \CQ_7 &=& \{\{p_1, p_2, p_6, p_7 \}\} \\
   \mbox{}
  \end{eqnarray*}
\end{minipage}

With $F = \{p_4, p_5\}$, for instance, processes~$p_1, p_2, p_3$ and $p_7$
are wise, $p_6$ is \naive, and $\CG_{\text{max}} = \{p_1, p_2, p_3\}$.  It
follows that process $p_7$ is wise but outside the guild
$\CG_{\text{max}}$, because the unique maximal quorum in $\CQ_7$ contains
the \naive process $p_6$.

Lemma~\ref{lem:quorumnaive} reveals the interesting result that for an
execution with a guild, each quorum of every correct process~$p_i$ contains
at least one process that is also in the maximal guild~$\CG_{\text{max}}$.
Since a kernel for $p_i$ is a process set that has some member in common
with every quorum of $p_i$, this implies that $\CG_{\text{max}}$ contains a
kernel for~$p_i$.

\begin{corollary}\label{cor:gmaxkernel}
  In every execution with a guild, the maximal guild~$\CG_{\text{max}}$
  contains a kernel for every correct process.
\end{corollary}

It follows that whenever all processes in the maximal guild send some
particular message, then every correct process will eventually receive this
message from all processes in one of its kernels.  This is exploited by
protocols that use kernels, such as Algorithm~\ref{alg:bracha} (in
Section~\ref{sec:broadcast}).

A guild can also be seen as a set of sufficiently many wise processes that
allow a protocol to make progress, in the following sense.

\begin{lemma}
  Consider an execution, in which the processes in $F$ are faulty and let
  $\mathcal{G}_{\text{max}}$ be the maximal guild for $F$.  Let $A$ be a
  superset of $F$ that is disjoint from $\mathcal{G}_{\text{max}}$, i.e.,
  $F \subseteq A \subseteq \CP \setminus \mathcal{G}_{\text{max}}$.
  
  Then, in any execution where the processes in $A$ fail,
  $\mathcal{G}_{\text{max}}$ is also the maximal guild for~$A$.
 \end{lemma}
 
\begin{proof}
  Let $\mathcal{G}_{\text{max}}$ be the maximal guild in an execution with
  set of faulty processes
  $F \subseteq \mathcal{P}\setminus \mathcal{G}_{\text{max}}$.  By
  definition of a guild, $\mathcal{G}_{\text{max}}$ contains a quorum for
  each of its members.  This means that there exists a quorum $Q_i$ for
  every $p_i \in \mathcal{G}_{\text{max}}$ such that
  $Q_i \cap~F = \emptyset.$ This also implies that for every set
  $A \supseteq F$, with
  $A \subseteq \mathcal{P} \setminus \mathcal{G}_{\text{max}}$, we have
  that $Q_i \cap~A = \emptyset$, and the lemma follows.
\end{proof}

Given the importance of a guild for an asymmetric Byzantine quorum system,
we introduce the following notion.

\begin{definition}[Tolerated system]\label{def:tolerated_sys}
  Given an asymmetric Byzantine quorum system \BQ and an execution with
  faulty processes $F$, a set of processes~$T$ is called \emph{tolerated
    (by \BQ)} if a non-empty guild \CG for $F$ and \BQ exists such that
  $T = \CP \setminus \CG$.

  The \emph{tolerated system} \CT of an asymmetric Byzantine quorum system
  \BQ is the maximal collection of tolerated sets, where $F$ ranges over
  all possible executions.
\end{definition}

Intuitively, the tolerated system of an asymmetric Byzantine quorum system
reflects its resilience: even when all processes in a tolerated set fail,
there still exists a non-empty guild.  Therefore, the tolerated system characterizes
the executions in which some processes will be able to operate correctly
and make progress (where progress is defined by the protocol they are
running). In that sense, the tolerated system of an asymmetric Byzantine
quorum system can be seen as a counterpart of the fail-prone system in the
symmetric model.

Notice that the tolerated system is a global notion emerging from the
subjective trust choices of the participating processes; any process that
knows the fail-prone and quorum systems of all processes can calculate it.
We remark that the tolerated system is a central concept for composing
asymmetric Byzantine quorum system, as shown by Alpos \emph{et
  al.}~\cite{DBLP:conf/srds/AlposCZ21}.

The following lemma shows that the tolerated system \CT of a canonical
asymmetric Byzantine quorum system is itself a symmetric fail-prone system.
In particular, $\tau$ builds a connection to symmetric quorum-based
protocols.  This property will be used in Section~\ref{sec:consensus} to
construct an asymmetric common coin protocol.

\begin{lemma}\label{lem:b3-q3}
  Let $\BQ$ be an asymmetric Byzantine quorum system among processes \CP with asymmetric fail-prone system
  $\BF = \overline{\BQ}$, i.e., such that \BQ is a canonical asymmetric Byzantine quorum system,
  and let \CT be the tolerated system of \BQ. If $B^3(\BF)$,
  then $Q^3(\CT)$.
\end{lemma}
\begin{proof}
  Towards a contradiction, let us assume that $\CT$ does not satisfy the $Q^3$-condition. This means that there exist $T_1, T_2, T_3 \in \CT$ such that $T_1 \cup T_2 \cup T_3 = \CP$. Also, let $\CG_1, \CG_2, \CG_3$ be the corresponding guilds, i.e., $\CG_1 = \CP \setminus T_1, \CG_2 = \CP \setminus T_2$ and $\CG_3 = \CP \setminus T_3$. 
By assumption, every guild contains at least one process and at least one quorum for this process is fully contained in the guild. By the consistency property of an asymmetric Byzantine quorum system, these quorums must intersect pairwise, hence the guilds also intersect pairwise. This means that there exist processes $p_{i} \in \CG_1 \cap \CG_2$ and $p_{j} \in \CG_2 \cap \CG_3$. Now, because $p_i$ is a member of $\CG_1$, we can make the following reasoning: $p_i$ has a quorum $Q_i \in \CQ_i$ such that $Q_i \subseteq \CG_1$, the quorum system is canonical, so $p_i$ has a fail-prone set $F_i = \CP \setminus Q_i \in \CF_i$, thus we get $ T_1 \subseteq F_i$, i.e., $T_1 \in \CF_i$. With similar reasoning, we get $T_2 \in \CF_i$ (because $p_i \in \CG_2$), $T_2 \in \CF_j$ (because $p_j \in \CG_2$), and $T_3 \in \CF_j$ (because $p_j \in \CG_3$). But this is a contradiction because $p_i$ and $p_j$ with fail-prone sets $T_1, T_2$, and $T_3$ violate the $B^3$-condition in $\BQ$.
\end{proof}

\section{Shared memory}
\label{sec:memory}

This section illustrates a first application of asymmetric quorum systems:
how to emulate shared memory, represented by a \emph{register}.
Maintaining a shared register reliably in a distributed system subject to
faults is perhaps the most fundamental task for which ordinary, symmetric
quorum systems have been introduced, in the models with
crashes~\cite{DBLP:conf/sosp/Gifford79} and with Byzantine
faults~\cite{DBLP:journals/dc/MalkhiR98}.


\subsection{Definitions}

\paragraph{Operations and precedence.}

For the particular \emph{shared-object} functionalities considered here,
the processes interact with an object $\Lambda$ through \emph{operations}
provided by $\Lambda$.  Operations on objects take time and are represented
by two events occurring at a process, an \emph{invocation} and a
\emph{response}.  The \emph{history} of an execution~$h$ consists of
the sequence of invocations and responses of $\Lambda$ occurring in
$h$.  An operation is \emph{complete} in a history if it has a
matching response.

An operation~$o$ \emph{precedes} another operation~$o'$ in a sequence
of events $h$, denoted $o <_h o'$, whenever $o$ completes
before $o'$ is invoked in $h$. A sequence of events $\pi$
\emph{preserves the real-time order} of a history~$h$ if for
every two operations $o$ and $o'$ in $\pi$, if $o <_h o'$ then
$o<_\pi o'$.  Two operations are \emph{concurrent} if neither one of
them precedes the other.  A sequence of events is \emph{sequential} if
it does not contain concurrent operations.
%
%
%
An execution on a shared object is \emph{well-formed} if the events at each
process are alternating invocations and matching responses, starting with
an invocation.

\paragraph{Semantics.}

A \emph{register} with domain~\CX provides two operations: $\op{write}(x)$,
which is parameterized by a value~$x \in \CX$ and outputs a token \str{ack}
when it completes; and $\op{read}$, which takes no parameter for invocation
but outputs a value $x \in \CX$ upon completion.

We consider a \emph{single-writer} (or \emph{SW}) register, where only a
designated process~$p_w$ may invoke~\op{write}, and permit \emph{multiple
  readers} (or \emph{MR}), that is, every process may execute a \op{read}
operation.  The register is initialized with a special value $x_0$, which
is written by an imaginary \op{write} operation that occurs before any
process invokes operations.  We consider \emph{regular} semantics under
concurrent access~\cite{DBLP:journals/dc/Lamport86}; the extension to other forms of
concurrent memory, including an atomic register, proceeds analogously.

It is customary in the literature to assume $p_w$ writes every value in \CX
at most once.  Furthermore, the writer and the reader are correct; with
asymmetric quorums we assume explicitly that readers and writers are
\emph{wise}.  We illustrate below why one cannot extend the guarantees of
the register to \naive processes.

\begin{definition}[Asymmetric Byzantine SWMR regular
  register]\label{def:byzreg}
  A protocol emulating an \emph{asymmetric SWMR regular register}
  satisfies:
  \begin{description}
  \item[Liveness:] If a wise process $p$ invokes an operation on the
    register, $p$ eventually completes the operation.
  \item[Safety:] Every \op{read} operation of a wise process that is not
    concurrent with a \op{write} returns the value written by the most
    recent, preceding \op{write} of a wise process; furthermore, a
    \op{read} operation of a wise process concurrent with a \op{write} of a
    wise process may also return the value that is written concurrently.
  \end{description}
\end{definition}

\subsection{Protocol with authenticated data}

In Algorithm~\ref{alg:regular-register-sigs}, we describe a protocol for
emulating a regular SWMR register with an asymmetric Byzantine quorum
system, for a designated writer $p_w$ and a reader $p_r \in \CP$.  The
protocol uses \emph{data authentication} implemented with digital
signatures.  This protocol is the same as the classic one of Malkhi and
Reiter~\cite{DBLP:journals/dc/MalkhiR98} that uses a Byzantine dissemination quorum system
and where processes send messages to each other over point-to-point links.
The difference lies in the individual choices of quorums by the processes
and that it ensures safety and liveness for wise processes.

In more detail, every process stores a triple $(\var{ts}, v, \sigma)$,
which consists of a timestamp~\ts, a value~$v$, and a signature~$\sigma$.
The idea is that the writer maintains a timestamp that increases with every
\op{write} operation.  The writer~$p_w$ signs the timestamp/value pair and
sends it in a message together with the signature to the processes, who
will store the data if the timestamp within the received message is higher
than the timestamp~\var{ts} stored locally.  A process then responds to
$p_w$ with an \str{ack} message.  The change from the classic protocol is
the writer~$p_w$ obtains \str{ack} messages from all processes in a quorum
$Q_w \in \CQ_w$ for itself.
The reader~$p_r$ sends a \str{read} message to all processes.  It then
waits to receive responses, which carry a triple of value, timestamp, and
signature such that the signature is valid, from processes in a
quorum~$Q_r$ for~$p_r$.  The returned value is the one from the triple with
the highest timestamp.

\begin{algo*}
\vbox{
\small
\begin{numbertabbing}\reset
  xxxx\=xxxx\=xxxx\=xxxx\=xxxx\=xxxx\=MMMMMMMMMMMMMMMMMMM\=\kill
  \textbf{State} \label{}\\
  \> \(\wts\): sequence number of write operations, stored only by writer~$p_w$ \label{}\\
  \> \(\rid\): identifier of read operations, used only by reader\label{} \\
  \> \(\ts, v, \sigma\): current state stored by $p_i$: timestamp, value, signature\label{} \\
  \\
  \textbf{upon invocation} \(\op{write}(v)\) \textbf{do}
  \` // only if $p_i$ is writer $p_w$ \label{}\\
  \> \(\wts \gets \wts + 1\) \label{}\\
  \> \(\sigma \gets \op{sign}_w(\str{write}\|w\|\wts\|v)\) \label{}\\
  \> send message \(\msg{write}{\wts, v, \sigma}\) to all \(p_j \in \CP\)\label{} \\
  \> \textbf{wait for} receiving a message \msg{ack}{} from all processes in some quorum \(Q_w \in \CQ_w\) \label{}\\
  \\
  \textbf{upon invocation} \(\op{read}\) \textbf{do}
  \` // only if $p_i$ is reader $p_r$ \label{}\\
  \> \(\rid \gets \rid + 1\) \label{}\\
  \> send message \(\msg{read}{\rid}\) to all \(p_j \in \CP\) \label{}\\
  \> \textbf{wait for} receiving messages \(\msg{value}{r_j, \ts_j, v_j, \sigma_j}\)
     from all processes in some \(Q_r \in \CQ_r\) \textbf{such that} \label{}\\
  \> \> \(r_j = \rid\) \textbf{and}
     \(\op{verify}_w(\sigma_j, \str{write}\|w\|\ts\|v_j)\) \label{}\\
  \> \textbf{return} \( \op{highestval}( \{ (\ts_j, v_j) | j \in Q_r\} ) \)\label{} \\
  \\
  \textbf{upon} receiving a message \(\msg{write}{\ts', v', \sigma'}\) from $p_w$ \textbf{do}
  \` // every process \label{}\\
  \> \textbf{if} \(\ts' > \ts\) \textbf{then} \label{}\\
  \> \> \((\ts, v, \sigma) \gets (\ts', v', \sigma')\) \label{}\\
  \> send message \(\msg{ack}{}\) to $p_w$ \label{}\\
  \\
  \textbf{upon} receiving a message \(\msg{read} r\) from $p_r$ \textbf{do} 
  \` // every process \label{}\\
  \> send message \(\msg{value}{r, \ts, v, \sigma}\) to $p_r$\label{}\\[-5ex]
\end{numbertabbing}
}
\caption{Emulation of an asymmetric SWMR regular register (process~$p_i$).}
\label{alg:regular-register-sigs}
\end{algo*}

The function $\op{highestval}(S)$ takes a set of timestamp/value pairs~$S$
as input and outputs the value in the pair with the largest timestamp,
i.e., $v$ such that $(\ts, v) \in S$ and
$\forall (\ts', v') \in S : \ts' < \ts \lor (\ts', v') = (\ts, v)$.  Note
that this $v$ is unique in Algorithm~\ref{alg:regular-register-sigs}
because $p_w$ is correct.
The protocol uses digital signatures, modeled by operations $\op{sign}_i$
and $\op{verify}_i$, as introduced earlier.


\begin{theorem}\label{thm:regular}
  Algorithm~\ref{alg:regular-register-sigs} emulates an asymmetric
  Byzantine SWMR regular register.
\end{theorem}

\begin{proof}
  First we show liveness for wise writer $p_w$ and reader $p_r$,
  respectively.  Since $p_w$ is wise by assumption, \(F \in {\CF_w}^*\),
  and by the availability condition of the quorum system there is
  \(Q_w \in \CQ_w\) with \(F \cap Q_w = \emptyset\). Therefore, the writer
  will receive sufficiently many \(\msg{ack}{}\) messages and the
  \op{write} will return.  As $p_r$ is wise, \(F \in {\CF_r}^*\), and by
  the analogous condition, there is \(Q_r \in \CQ_r\) with
  \(F \cap Q_r = \emptyset\). Because $p_w$ is correct and by the
  properties of the signature scheme, all responses from processes
  \(p_j \in Q_r\) satisfy the checks and \op{read} returns.

  Regarding safety, it is easy to observe that any value output by
  \op{read} has been written in some preceding or concurrent \op{write}
  operation, and this even holds for \naive readers and writers. This
  follows from the properties of the signature scheme; \op{read} verifies
  the signature and outputs only values with a valid signature produced by~
  $p_w$.

  We now argue that when both the writer and the reader are wise, then
  \op{read} outputs a value of either the last preceding \op{write} or a
  concurrent \op{write} and the protocol satisfies safety for a regular
  register. On a high level, note that \(F \in {\CF_w}^* \cap {\CF_r}^*\)
  since both are wise. So if $p_w$ writes to a quorum \(Q_w \in \CQ_w\) and
  $p_r$ reads from a quorum \(Q_r \in \CQ_r\), then by consistency of the
  quorum system \(Q_w \cap Q_r \not\subseteq F\) because $p_w$ and $p_r$
  are wise.  Hence, there is some correct \(p_i \in Q_w \cap Q_r\) that
  received the most recently written value from $p_w$ and returns it
  to~$p_r$.
\end{proof}

\begin{example} 
We show why the guarantees of this protocol with
asymmetric quorums hold only for wise readers and writers.  Consider
$\BQ_A$ from the last section and an execution in which $p_2$ and $p_4$ are
faulty, and therefore $p_1$ is \naive and $p_3$ and $p_5$ are wise.  A
quorum for $p_1$ consists of $p_1$ and three processes in
$\{ p_2, \dots, p_5\}$; moreover, every process set that contains $p_3$,
one of $\{p_1, p_2\}$ and one of $\{p_4, p_5\}$ is a quorum for~$p_3$.

We illustrate that if \naive $p_1$ writes, then a wise reader $p_3$ may
violate safety.  Suppose that all correct processes, especially $p_3$,
store timestamp/value/signature triples from an operation that has
terminated and that wrote~$x$.  When $p_1$ invokes $\op{write}(u)$, it
obtains \msg{ack}{} messages from all processes except~$p_3$.  This is a
quorum for~$p_1$.  Then $p_3$ runs a \op{read} operation and receives the
outdated values representing $x$ from itself ($p_3$ is correct but has not
been involved in writing $u$) and also from the faulty $p_2$ and $p_4$.
Hence, $p_3$ outputs $x$ instead of~$u$.

Analogously, with the same setup of every process initially storing a
representation of $x$ but with wise $p_3$ as writer, suppose $p_3$ executes
$\op{write}(u)$.  It obtains \msg{ack}{} messages from $p_2$, $p_3$, and
$p_4$ and terminates.  When $p_1$ subsequently invokes \op{read} and
receives values representing $x$, from correct $p_1$ and $p_5$ and from
faulty $p_2$ and $p_4$, then $p_1$ outputs $x$ instead of~$y$ and violates
safety as a \naive reader.

Since the sample operations are not concurrent, the implication actually
holds also for registers with only safe semantics.
\end{example}

\subsection{Double-write protocol without data authentication}

\newcommand\pts{\var{pts}\xspace}
\newcommand\pv{\var{pv}\xspace}

This section describes a second protocol emulating an asymmetric Byzantine
SWMR regular register.  In contrast to the previous protocol, it does not
use digital signatures for authenticating the data to the reader.  Our
algorithm generalizes the construction of Abraham et
al.~\cite{DBLP:journals/dc/AbrahamCKM06} and also assumes that only a
finite number of write operations occur (\emph{FW-termination}).
Furthermore, this algorithm illustrates the use of asymmetric core-set
systems in the context of an asymmetric-trust protocol.

This protocol extends Algorithm~\ref{alg:regular-register-sigs} and every
process stores the most recently written timestamp-value pair $(\ts, v)$.
Every \op{write} operation performs two rounds instead of one, a pre-write
round and a write round.  In addition to the previous protocol, every
process stores the most recently pre-written timestamp-value
pair~$(\pts, \pv)$.  From the perspective of the writer~$p_w$, each round
proceeds like the single round in
Algorithm~\ref{alg:regular-register-sigs}, except that $p_w$ does not
produce a digital signature.  In particular, $p_w$ waits in each round for
responses that form a quorum $Q_w \in \CQ_w$ for itself.

The reader~$p_r$ exchanges one round of messages with the processes and
waits for responses that form a quorum $Q_r \in \CQ_r$ for~$p_r$.  Every
response contains the pre-written and the written timestamp-value pairs
from the sending process.  The reader collects these in an array
\var{readlist} until the following condition is satisfied.  A pair
$(\ts^*, v^*)$, a core set $C_r$ for $p_r$ of entries in \var{readlist},
and a quorum $Q_r$ for $p_r$ of entries in \var{readlist} exist such that
(1) the pair $(\ts^*, v^*)$ is either the pre-written or the written pair
in all entries of \var{readlist} in $C_r$; and (2) $(\ts^*, v^*)$ is the
pair with the highest timestamp among the entries in $Q_r$.  Intuitively,
the initial pre-write round and the core set~$C_r$ that reports this value
to $p_r$ replace the step of authenticating the value through a digital
signature.  This respects safety because $C_r$, for a wise $p_r$, contains
at least one correct process that has not altered the value.  The full
protocol appears in Algorithm~\ref{alg:regular-register-doublewrite}.

\begin{algo*}
\vbox{
\small
\begin{numbertabbing}\reset
  xxxx\=xxxx\=xxxx\=xxxx\=xxxx\=xxxx\=MMMMMMMMMMMMMMMMMMM\=\kill
  \textbf{State} \label{} \\
  \> \(\wts\): sequence number of write operations, stored only by writer~$p_w$ \label{}\\
  \> \(\rid\): identifier of read operations, used only by reader \label{}\\
  \> \(\pts, \pv, \ts, v\): current state stored by $p_i$: pre-written timestamp and value, written timestamp and value \label{}\\
  \\
  \textbf{upon invocation} \(\op{write}(v)\) \textbf{do}
  \` // only if $p_i$ is writer $p_w$ \label{}\\
  \> \(\wts \gets \wts + 1\)  \label{}\\
  \> send message \(\msg{prewrite}{\wts, v}\) to all \(p_j \in \CP\) \label{}\\
  \> \textbf{wait for} receiving a message \msg{preack}{} from all processes
  in some quorum \(Q_w \in \CQ_w\) \label{}\\
  \> send message \(\msg{write}{\wts, v}\) to all \(p_j \in \CP\) \label{}\\
  \> \textbf{wait for} receiving a message \msg{ack}{} from all processes
  in some quorum \(Q_w \in \CQ_w\) \label{}\\
  \\
  \textbf{upon invocation} \op{read} \textbf{do}
  \` // only if $p_i$ is reader $p_r$ \label{}\\
  \> \(\rid \gets \rid + 1\) \label{}\\
  \> send message \(\msg{read}{\rid}\) to all \(p_j \in \CP\) \label{}\\
  \\
  \textbf{upon} receiving a message
     \([\str{value}, r_j, \pts_j, \pv_j, \ts_j, v_j]\) from $p_j$
     \textbf{such that}
  \` // only if $p_i$ is reader $p_r$ \label{}\\
  \> \> \(r_j = \rid \land \bigl(\pts_j = \ts_j + 1 \vee
     (\pts_j,\pv_j) = (\ts_j, v_j)\bigr)\) \textbf{do} \label{} \\
  \> $\var{readlist}[j] \gets (\pts_j, \pv_j, \ts_j, v_j)$ \label{}\\
  \> \textbf{if} there exist $\ts^*,v^*$,
     a core set $C_r \in \CC_r$ for $p_r$,
     and a quorum $Q_r \in \CQ_r$ for $p_r$ \textbf{such that}\label{} \\
  \> \> $C_r \subseteq \bigl\{p_k |
     \var{readlist}[k] = (\pts_k, \pv_k, \ts_k, v_k) \} \land \bigl(
     (\pts_k, \pv_k) = (\ts^*, v^*) \lor (\ts_k, v_k) = (\ts^*, v^*) \bigr)
     \bigr\}$ \textbf{and}\label{} \\
  \> \> $Q_r = \bigl\{p_k |
        \var{readlist}[k] = (\pts_k, \pv_k, \ts_k, v_k)$ \label{}\\
  \> \> \> \> $\mbox{} \land \bigl( (\ts_k < \ts^*) \lor
        (\pts_k, \pv_k) = (\ts^*, v^*) \lor (\ts_k, v_k) = (\ts^*, v^*) \bigr)
        \bigr\}$ \textbf{then} \label{}\\
  \> \> \textbf{return} \(v^*\) \label{}\\
  \> \textbf{else} \label{}\\
  \> \> send message \msg{read}{rid} to all \(p_j \in \CP\)\label{}\\
  \\
  \textbf{upon} receiving a message \(\msg{prewrite}{\ts', v'}\) from $p_w$
    \textbf{such that} \(\ts' = \pts + 1 \land \pts = \ts\) \textbf{do} \label{}\\
    \> \((\pts, \pv) \gets (\ts', v')\) \label{}\\
    \> send message \(\msg{preack}{}\) to $p_w$\label{} \\
  \\
  \textbf{upon} receiving a message \(\msg{write}{\ts', v'}\) from $p_w$
    \textbf{such that} \(\ts' = \pts \land v' = \pv\) \textbf{do} \label{}\\
    \> \((\ts, v) \gets (\ts', v')\) \label{}\\
    \> send message \(\msg{ack}{}\) to $p_w$ \label{}\\
  \\
  \textbf{upon} receiving a message \(\msg{read}{r}\) from $p_r$ \textbf{do} \label{}\\
    \> send message \(\msg{value}{r, \pts, \pv, \ts, v}\) to $p_r$\label{}\\[-5ex]
\end{numbertabbing}
}
\caption{Double-write emulation of an asymmetric SWMR regular register (process~$p_i$).}
\label{alg:regular-register-doublewrite}
\end{algo*}

\begin{theorem}\label{thm:doublewrite}
  Algorithm~\ref{alg:regular-register-doublewrite} emulates an asymmetric
  Byzantine SWMR regular register, provided there are only finitely many
  write operations.
\end{theorem}

\begin{proof}
  We first establish safety when the writer $p_w$ and the reader $p_r$ are
  wise. In that case, \(F \in {\CF_w}^* \cap {\CF_r}^*\).
  During in a \op{write} operation, $p_w$ has received \str{preack} and
  \str{ack} messages from $Q_w \in \CQ_i$ and $Q_w' \in \CQ_i$,
  respectively, and for all \(Q_r \in \CQ_r\) it holds that
  \(Q_w\cap Q_r \not\subseteq F\) and \(Q_w'\cap Q_r\not\subseteq F\).

  We now argue that any pair \((\ts^*,v^*)\) returned by $p_r$ was written
  by $p_w$ either in a preceding or a concurrent \op{write}.  From the
  properties of the core set $C_r$, because $p_r$ is wise, and together
  with the condition that $(\ts^*, v^*)$ satisfies, it follows that at
  least one correct process exists in $C_r$ that stores $(\ts^*, v^*)$ as a
  pre-written or as a written value.  Thus, the pair was written by $p_w$
  before.

  Next we argue that for every completed \(\op{write}(v^*)\) operation, in
  which $p_w$ has sent \msg{write}{\wts, v^*}, and for any subsequent
  \op{read} operation that selects \((\ts^*,v^*)\) and returns $v^*$, it
  must hold \(\wts \leq \ts^*\).  Namely, the condition on $Q_r$ implies
  that \(\ts^* \geq \ts_k\) for all \(p_k \in Q_r\). By the consistency of
  the quorum system, it holds that \(Q_w' \cap Q_r \not\subseteq F\), so
  there is a correct process \(p_\ell \in Q_w'\cap Q_r\) that has sent
  $\ts_\ell$ to $p_r$.  Then \(\ts^* \geq \ts_\ell \geq \wts\) follows
  because the timestamp variable of \(p_\ell\) only increases.

  The combination of the above two paragraphs implies that for \op{read}
  operations that are not concurrent with any \op{write}, the pair
  \((\ts^*,v^*)\) chosen by \op{read} was actually written in the
  immediately preceding \op{write}.  If the \op{read} operation occurs
  concurrently with a \op{write}, then the pair \((\ts^*,v^*)\) chosen by
  \op{read} may also originate from the concurrent \op{write}.  This
  establishes the safety property of the SWMR regular register.

  We now show liveness. First, if $p_w$ is wise, then there exists a quorum
  \(Q_w \in \CQ_w\) such that \(Q_w \cap F = \emptyset\).  Second, any
  correct process will eventually receive all \msg{prewrite}{\wts,v} and
  \msg{write}{\wts,v} messages sent by $p_w$ and process them in the
  correct order by the assumption of FIFO links. This means that $p_w$ will
  receive \msg{preack}{} and \msg{ack}{} messages, respectively, from all
  processes in one of its quorums, since at least the processes in $Q_w$
  will eventually send those.

  Liveness for the reader $p_r$ is shown under the condition that $p_r$ is
  wise and that the \op{read} operation is concurrent with only finitely
  many \op{write} operations. The latter condition implies that there is
  one last \op{write} operation that is initiated, but does not necessarily
  terminate, while \op{read} is active.

  By the assumption that $p_w$ is correct and because messages are received
  in FIFO order, all messages of that last \op{write} operation will
  eventually arrive at the correct processes.  Notice also that $p_r$
  simply repeats its steps until it succeeds and returns a value that
  fulfills the condition.  Hence, there is a time after which all correct
  processes reply with \str{value} messages that contain pre-written and
  written timestamp/value pairs from that last operation.  It is easy to
  see that there exist a core set and a quorum for $p_r$ that satisfy the
  condition and the reader returns.  In conclusion, the algorithm emulates
  an asymmetric regular SWMR register, where liveness holds only for
  finitely many write operations.
\end{proof}

\section{Broadcast}
\label{sec:broadcast}

This section shows how to implement two \emph{broadcast primitives}
tolerating Byzantine faults with asymmetric quorums.  Recall from the
standard literature~\cite{10.5555/HadzilacosT93,CharronBostPS10,DBLP:books/daglib/0025983} that
reliable broadcasts offer basic forms of reliable message delivery and
consistency, but they do not impose a total order on delivered messages (as
this is equivalent to consensus).  The Byzantine broadcast primitives
described here, \emph{consistent broadcast} and \emph{reliable broadcast},
are prominent building blocks for many more advanced protocols.

With both primitives, the sender process may broadcast a message $m$ by
invoking $\op{broadcast}(m)$; the broadcast abstraction outputs $m$ to the
local application on the process through a $\op{deliver}(m)$ event.
Moreover, the notions of broadcast considered in this section are intended
to deliver only one message per instance.  Every instance has a distinct
(implicit) label and a designated sender~$p_s$.  With standard multiplexing
techniques one can extend this to a protocol in which all processes may
broadcast messages repeatedly~\cite{DBLP:books/daglib/0025983}.

\paragraph{Byzantine consistent broadcast.}
The simplest such primitive, which has been called \emph{(Byzantine)
  consistent broadcast}~\cite{DBLP:books/daglib/0025983}, ensures only that those correct
processes which deliver a message agree on the content of the message, but
they may not agree on termination.  In other words, the primitive does not
enforce ``reliability'' such that a correct process outputs a message if
and only if all other correct processes produce an output.  The events in
its interface are denoted by \op{c-broadcast} and \op{c-deliver}.

The change of the definition towards asymmetric quorums affects most of its
guarantees, which hold only for wise processes but not for all correct
ones.  This is similar to the definition of a register in
Section~\ref{sec:memory}.

\begin{definition}[Asymmetric Byzantine consistent
  broadcast]\label{def:acbc}
  A protocol for \emph{asymmetric (Byzantine) consistent broadcast} satisfies:

\begin{description}
\item[Validity:] If a correct process $p_s$ \op{c-broadcasts} a
  message~$m$, then all wise processes eventually \op{c-deliver}~$m$.
  
\item[Consistency:] If some wise process \op{c-delivers}~$m$ and another
  wise process \op{c-delivers}~$m'$, then $m=m'$.
  
\item[Integrity:] For any message~$m$, every correct process
  \op{c-delivers} $m$ at most once.  Moreover, if the sender $p_s$ is
  correct and the receiver is wise, then $m$ was previously
  \op{c-broadcast} by $p_s$.
\end{description}

\end{definition}

The following protocol is an extension of ``authenticated echo
broadcast''~\cite{DBLP:books/daglib/0025983}, which goes back to Srikanth
and Toueg~\cite{DBLP:journals/dc/SrikanthT87}.  It is a building block
found in many Byzantine fault-tolerant protocols with greater complexity.
The protocol first has the sender~$p_s$ send its message~$m$ to all
processes; then every process echoes $m$, in the sense that it rebroadcasts
an \str{echo} message with $m$ to all processes.  As soon as a process
receives a quorum of such \str{echo} messages that all contain the
same~$m'$, the process \op{c-delivers}~$m'$.
The adaptation for asymmetric quorums is straightforward: Every process
considers its own quorum system before \op{c-delivering} the message.

\begin{algo*}
\vbox{
\small
\begin{numbertabbing}\reset
  xxxx\=xxxx\=xxxx\=xxxx\=xxxx\=xxxx\=MMMMMMMMMMMMMMMMMMM\=\kill
  \textbf{State} \label{}\\
  \> \(\var{sentecho} \gets \false\): indicates whether $p_i$ has sent \str{echo} \label{}\\
  \> \(\var{echos} \gets [\bot]^N\): collects the received \str{echo} messages from other processes \label{}\\
  \> \(\var{delivered} \gets \false\): indicates whether $p_i$ has delivered a message\label{}\\
  \\
  \textbf{upon invocation} \(\op{c-broadcast}(m)\) \textbf{do} \label{}\\
  \> send message \msg{send}{m} to all \(p_j \in \CP\) \label{}\\
  \\
  \textbf{upon} receiving a message \msg{send}{m} from $p_s$
     \textbf{such that} \(\neg \var{sentecho}\) \textbf{do} \label{}\\
  \> \(\var{sentecho} \gets \true\)\label{} \\
  \> send message \msg{echo}{m} to all \(p_j \in \CP\) \label{}\\
  \\
  \textbf{upon} receiving a message \msg{echo}{m} from \(p_j\) \textbf{do} \label{}\\
  \> \textbf{if} \(\var{echos}[j] = \bot\) \textbf{then}\label{} \\
  \> \> \(\var{echos}[j] \gets m\) \label{}\\
  \\
  \textbf{upon exists} \(m \not= \bot\) \textbf{such that}
     \(\{ p_j \in \CP | \var{echos}[j] = m\} \in \CQ_i\) \textbf{and}
     \(\neg \var{delivered}\) \textbf{do} \label{}\\
  \> \(\var{delivered} \gets \true\) \label{}\\
  \> \textbf{output} \(\op{c-deliver}(m)\)\label{}\\[-5ex]
\end{numbertabbing}
}
\caption{Asymmetric Byzantine consistent broadcast protocol with sender~$p_s$
  (process~$p_i$)}
\label{alg:srikanth-toueg}
\end{algo*}

\begin{theorem}\label{thm:srikanth-toueg}
  Algorithm~\ref{alg:srikanth-toueg} implements asymmetric Byzantine
  consistent broadcast.
\end{theorem}

\begin{proof}
  For the \emph{validity} property, it is straightforward to see that every
  correct process sends \msg{echo}{m}.  According to the availability
  condition for the quorum system~$\CQ_i$ of every wise process~$p_i$ and
  because $F \subseteq F_i$ for some $F_i \in \CF_i$, there exists some
  quorum $Q_i$ for $p_i$ of correct processes that echo $m$ to $p_i$.
  Hence, $p_i$ \op{c-delivers}~$m$.

  To show \emph{consistency}, suppose that some wise process~$p_i$ has
  \op{c-delivered} $m_i$ because of \msg{echo}{m_i} messages from a
  quorum~$Q_i$ and another wise~$p_j$ has received \msg{echo}{m_j} from all
  processes in $Q_j \in \CQ_j$.  By the consistency property of \BQ it
  holds \(Q_i \cap Q_j \not\subseteq F\); let $p_k$ be this process in
  $Q_i \cap Q_j$ that is not in $F$.  Because $p_k$ is correct, $p_i$ and
  $p_j$ received the same message from $p_k$ and $m_i = m_j$.

  The first condition of \emph{integrity} is guaranteed by using the
  \var{delivered} flag; the second condition holds because because the
  receiver is wise, and therefore the quorum that it uses for the decision
  contains some correct processes that have sent \msg{echo}{m} with the
  message~$m$ they obtained from~$p_s$ according to the protocol.
\end{proof}

\begin{example}%
We illustrate the broadcast protocols using a six-process asymmetric quorum
system $\BQ_C$, defined through its fail-prone system~$\BF_C$.  In $\BF_C$,
as shown below, for $p_1$, $p_2$, and $p_3$, each process always trusts
itself, some other process of $\{p_1, p_2, p_3\}$ and one further process
in $\{p_1, \dots, p_5\}$.  Process~$p_4$ and $p_5$ each assumes that at
most one other process of $\{p_1, \dots, p_5\}$ may fail (excluding
itself).  Moreover, none of the processes $p_1$, \dots, $p_5$ ever
trusts~$p_6$.  For $p_6$ itself, the fail-prone set is $\{p_1, p_3\}$,
i.e., it trusts $p_2$, $p_4$, and $p_5$ unconditionally.

\vspace*{-1ex}
\begin{minipage}[c]{0.1\linewidth}
  \vspace*{2ex}
  \center\Large$\BF_C$:
\end{minipage}
\begin{minipage}[l]{0.4\linewidth}
  \begin{eqnarray*}
    \mbox{}\\
    \CF_1 & = & \Theta^3_2(\{p_2,p_4,p_5\}) \ast \{\{p_6\}\} \\
    \CF_2 & = & \Theta^3_2(\{p_3,p_4,p_5\}) \ast \{\{p_6\}\} \\
    \CF_3 & = & \Theta^3_2(\{p_1,p_4,p_5\}) \ast \{\{p_6\}\} \\
    \CF_4 & = & \Theta^4_1(\{p_1,p_2,p_3,p_5\}) \ast \{\{p_6\}\} \\
    \CF_5 & = & \Theta^4_1(\{p_1,p_2,p_3,p_4\}) \ast \{\{p_6\}\} \\
    \CF_6 & = & \{\{p_1, p_3\}\}  \\
    \mbox{}
  \end{eqnarray*}
\end{minipage}
\begin{minipage}[l]{0.5\linewidth}
  \includegraphics[scale=0.42]{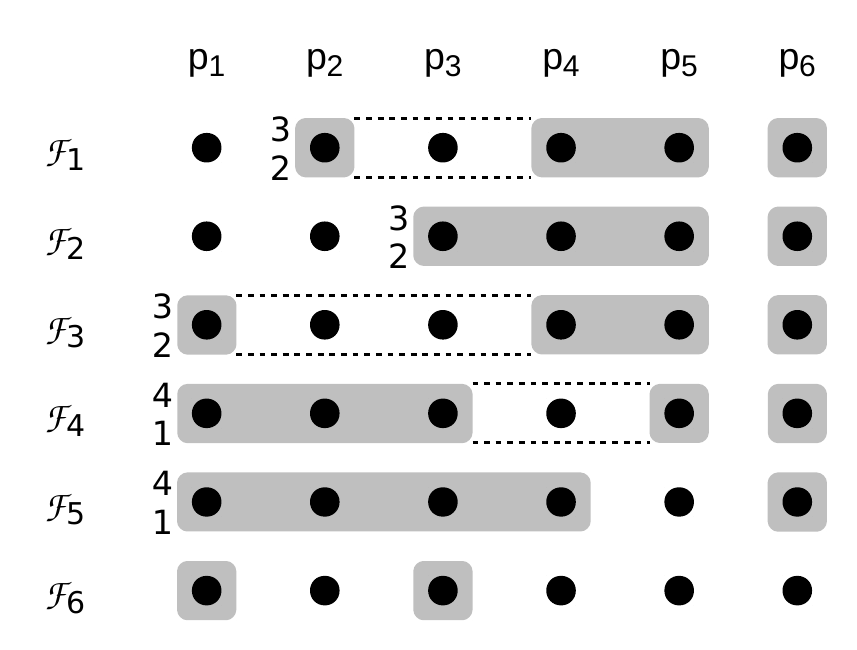}
\end{minipage}

\begin{minipage}[c]{0.1\linewidth}
  \vspace*{2ex}
  \center\Large$\BQ_C$:
\end{minipage}
\begin{minipage}[l]{0.5\linewidth}
  \begin{eqnarray*}
    \mbox{}\\
   \CQ_1 &=& \{\{p_1, p_3, p_5 \}, \{p_1, p_3, p_4 \}, \{p_1, p_2, p_3 \}\} \\
   \CQ_2 &=& \{\{p_1, p_2, p_5 \}, \{p_1, p_2, p_4 \}, \{p_1, p_2, p_3 \}\} \\
   \CQ_3  &=& \{\{p_2, p_3, p_5 \}, \{p_2, p_3, p_4 \}, \{p_1, p_2, p_3 \}\} \\
  \CQ_4  &=& \{\{p_1, p_2, p_3, p_4 \}, \{p_1, p_2, p_4, p_5 \},
              \{p_1, p_3, p_4, p_5 \}, \{p_2, p_3, p_4, p_5 \}\} \\
   \CQ_5  &=& \{\{p_1, p_2, p_3, p_5 \}, \{p_1, p_2, p_4, p_5 \},
               \{p_1, p_3, p_4, p_5 \}, \{p_2, p_3, p_4, p_5 \}\} \\
   \CQ_6 &=& \{\{p_3, p_4, p_5, p_6\}, \{p_2, p_4, p_5, p_6\},
             \{p_2, p_3, p_5, p_6\}, \{p_2, p_3, p_4, p_6 \}\} \\
 \end{eqnarray*}
 \end{minipage}

One can verify that $B^3(\BF_C)$ holds; hence, let $\BQ_C$ be the canonical
quorum system of $\BF_C$.  Again, there is no reliable process that could
be trusted by all and $\BQ_C$ is not a special case of a symmetric
threshold Byzantine quorum system.  With $F = \{p_1, p_5\}$, for instance,
process~$p_3$ is wise, $p_2$, $p_4$, and $p_6$ are \naive, and there is no
guild.

Consider now an execution of Algorithm~\ref{alg:srikanth-toueg} with sender
$p_4^*$ and $F = \{p_4^*, p_5^*\}$ (we write $p_4^*$ and $p_5^*$ to denote
that they are faulty).  This means processes~$p_1, p_2, p_3$ are wise and
form a guild because $\{p_1, p_2, p_3 \}$ is a quorum for all three;
furthermore, $p_6$ is \naive.  The protocol execution proceeds as follows,
with steps from left to right:

\begin{center}
  \small\vspace*{-4ex}
\begin{align*}
  &  & p_1  &: \:\: \msg{echo}{x} \to \CP   & p_1 &: \op{c-deliver}(x) \\
  &  & p_2  &: \:\: \msg{echo}{u} \to \CP   & p_2 &: \text{no quorum of \msg{echo}{} in $\CQ_2$} \\
  &  & p_3  &: \:\: \msg{echo}{x} \to \CP   & p_3 &: \text{no quorum of \msg{echo}{} in $\CQ_3$} \\
  p_4^* & : \begin{cases}\msg{send}{x} \to p_1, p_3\\\msg{send}{u} \to p_2, p_6\end{cases}
     & p_4^*&: \begin{cases}\msg{echo}{x} \to p_1\\ \msg{echo}{u} \to p_6\end{cases} \\
  &  & p_5^*&: \begin{cases}\msg{echo}{x} \to p_1\\ \msg{echo}{u} \to p_6\end{cases} \\
  &  & p_6  &: \:\: \msg{echo}{u} \to \CP   & p_6  &: \op{c-deliver}(u)
\end{align*}
\end{center}
Hence, $p_1$ receives $\msg{echo}{x}$ from, say,
$\{p_1, p_3, p_4^*\} \in \CQ_1$ and \op{c-delivers} $x$, but the
other wise processes do not terminate.  The \naive $p_6$ gets
$\msg{echo}{u}$ from $\{p_2, p_4^*, p_5^*, p_6\} \in \CQ_6$ and
\op{c-delivers}~$u \neq x$.
\end{example}

\paragraph{Byzantine reliable broadcast.}
In the symmetric setting, consistent broadcast has been extended to
\emph{(Byzantine) reliable broadcast} in a well-known way to address the
disagreement about termination among the correct
processes~\cite{DBLP:books/daglib/0025983}.  This primitive has the same interface as
consistent broadcast, except that its events are called \op{r-broadcast}
and \op{r-deliver} instead of \op{c-broadcast} and \op{c-deliver},
respectively.

A reliable broadcast protocol also has all properties of consistent
broadcast, but satisfies the additional \emph{totality} property stated
next.  Taken together, \emph{consistency} and \emph{totality} imply a
notion of \emph{agreement}, similar to what is also ensured by many
crash-tolerant broadcast primitives.  Analogously to the earlier primitives
with asymmetric trust, our notion of an \emph{asymmetric reliable
  broadcast}, defined next, ensures agreement on termination only for the
wise processes, and moreover only for executions with a guild.  Also the
\emph{validity} of Definition~\ref{def:acbc} is extended by the assumption
of a guild.  Intuitively, one needs a guild because the wise processes that
make up the guild are self-sufficient, in the sense that the guild contains
a quorum of wise processes for each of its members; without that, there may
not be enough wise processes.

\begin{definition}[Asymmetric Byzantine reliable broadcast]\label{def:arbc}
  A protocol for \emph{asymmetric (Byzantine) reliable broadcast} is a
  protocol for asymmetric Byzantine consistent broadcast with the revised
  \emph{validity} condition and the additional \emph{totality} condition
  stated next:
\begin{description}
\item[Validity:] In all executions with a guild, if a correct process $p_s$
  \op{r-broadcasts} a message~$m$, then all processes in the maximal guild
  eventually \op{r-deliver}~$m$.
\item[Totality:] In all executions with a guild, if a wise process
  \op{r-delivers} some message, then all processes in the maximal guild
  eventually \op{r-deliver} a message.
\end{description}
\end{definition}

The protocol of Bracha~\cite{DBLP:journals/iandc/Bracha87} implements reliable
broadcast subject to Byzantine faults with symmetric trust.  It augments
the authenticated echo broadcast from
Algorithm~\ref{alg:srikanth-toueg} with a second all-to-all exchange, where
each process is supposed to send \str{ready} with the payload message that
will be \op{r-delivered}.  When a process receives the same $m$ in $2f+1$
\str{ready} messages, in the symmetric model with a threshold Byzantine
quorum system, then it \op{r-delivers}~$m$.
Also, a process that receives \msg{ready}{m} from $f+1$ distinct
processes and that has not yet sent a \str{ready} chimes in and also sends
\msg{ready}{m}.  These two steps ensure totality.

For asymmetric quorums, the conditions of a process~$p_i$ receiving $f+1$
and $2f+1$ equal \str{ready} messages, respectively, generalize to
receiving the same message from a kernel for $p_i$ and from a quorum for
$p_i$.  Intuitively, the change in the first condition ensures that when a
wise process $p_i$ (that is also in the maximal guild) receives the same \msg{ready}{m} message from a kernel
for itself, then this kernel intersects with some quorum of wise processes.
Therefore, at least one wise process has sent \msg{ready}{m} and $p_i$ can
safely adopt $m$.
Furthermore, the change in the second condition relies on the properties
of asymmetric quorums to guarantee that
whenever some wise process has \op{r-delivered}~$m$, then
enough correct processes have sent a \msg{ready}{m} message such that all
wise processes eventually receive a kernel of \msg{ready}{m} messages and
also send~\msg{ready}{m}.

\begin{algo*}[tbh]
\vbox{
\small
\begin{numbertabbing}\reset
  xxxx\=xxxx\=xxxx\=xxxx\=xxxx\=xxxx\=MMMMMMMMMMMMMMMMMMM\=\kill
  \textbf{State} \label{}\\
  \> \(\var{sentecho} \gets \false\): indicates whether $p_i$ has sent \str{echo} \label{}\\
  \> \(\var{echos} \gets [\bot]^N\): collects the received \str{echo} messages from other processes \label{}\\
  \> \(\var{sentready} \gets \false\): indicates whether $p_i$ has sent \str{ready} \label{}\\
  \> \(\var{readys} \gets [\bot]^N\): collects the received \str{ready} messages from other processes \label{}\\
  \> \(\var{delivered} \gets \false\): indicates whether $p_i$ has delivered a message\label{}\\
  \\
  \textbf{upon invocation} \(\op{r-broadcast}(m)\) \textbf{do} \label{}\\
  \> send message \msg{send}{m} to all \(p_j \in \CP\) \label{}\\
  \\
  \textbf{upon} receiving a message \msg{send}{m} from $p_s$
    \textbf{such that} \(\neg \var{sentecho}\) \textbf{do} \label{}\\
  \> \(\var{sentecho} \gets \true\) \label{}\\
  \> send message \msg{echo}{m} to all \(p_j \in \CP\)\label{} \\
  \\
  \textbf{upon} receiving a message \msg{echo}{m} from \(p_j\) \textbf{do} \label{}\\
  \> \textbf{if} \(\var{echos}[j] = \bot\) \textbf{then} \label{}\\
  \> \> \(\var{echos}[j] \gets m\) \label{}\\
  \\
  \textbf{upon exists} \(m \not= \bot \) \textbf{such that}
     \( \{p_j \in \CP | \var{echos}[j] = m\} \in \CQ_i \)
     \textbf{and}  \(\neg \var{sentready}\) \textbf{do}
     \` // a quorum for $p_i$ \label{}\\
  \> \(\var{sentready} \gets \true\) \label{}\\
  \> send message \msg{ready}{m} to all $p_j \in \CP$ \label{}\\
  \\
  \textbf{upon exists} \(m \not= \bot \) \textbf{such that}
     \( \{p_j \in \CP | \var{readys}[j] = m\} \in \CK_i \)
     \textbf{and} \(\neg \var{sentready}\) \textbf{do}
     \` // a kernel for $p_i$\label{} \\
  \> \(\var{sentready} \gets \true\) \label{}\\
  \> send message \msg{ready}{m} to all $p_j \in \CP$ \label{}\\
  \\
  \textbf{upon} receiving a message \msg{ready}{m} from \(p_j\)
     \textbf{do} \label{}\\
  \> \textbf{if} \(\var{readys}[j] = \bot\) \textbf{then} \label{}\\
  \> \> \(\var{readys}[j] \gets m\) \label{}\\
  \\
  \textbf{upon exists} \(m \not= \bot \) \textbf{such that}
     \( \{ p_j \in \CP | \var{readys}[j] = m\} \in \CQ_i \) 
     \textbf{and} \(\neg \var{delivered}\) \textbf{do} \label{}\\
  \> \(\var{delivered} \gets \true\) \label{}\\
  \> \textbf{output} \(\op{r-deliver(m)}\)\label{}\\[-5ex]
\end{numbertabbing}
}
\caption{Asymmetric Byzantine reliable broadcast protocol with sender~$p_s$
  (process~$p_i$)}
\label{alg:bracha}
\end{algo*}

Applying these changes to Bracha's protocol results in the asymmetric
reliable broadcast protocol shown in Algorithm~\ref{alg:bracha}.  Note that
it strictly extends Algorithm~\ref{alg:srikanth-toueg} by the additional
round of \str{ready} messages, in the same way as for symmetric trust.  For
instance, when instantiated with the symmetric threshold quorum system of
$n = 3f + 1$ processes, of which $f$ may fail, then every set of $f + 1$
processes is a kernel.

In Algorithm~\ref{alg:bracha}, there are two conditions that let a correct
$p_i$ send \msg{ready}{m}: either receiving a quorum of \msg{echo}{m}
messages for itself or obtaining a kernel for itself of \msg{ready}{m}
messages.  For the first case, we say $p_i$ \emph{sends \str{ready} after
  \str{echo}}; for the second case, we say $p_i$ \emph{sends \str{ready}
  after \str{ready}}.

\begin{lemma}\label{lem:brachaunique}
  In any execution with a guild, there exists a unique $m$ such that
  whenever a wise process in the maximal guild sends a \str{ready} message,
  it contains~$m$.
\end{lemma}

\begin{proof}
  Consider first all \str{ready} messages sent by wise processes
  after~\str{echo}.  The fact that Algorithm~\ref{alg:bracha} extends
  Algorithm~\ref{alg:srikanth-toueg} achieving consistent broadcast,
  combined with the consistency property in Definition~\ref{def:acbc}
  implies immediately that the lemma holds for \str{ready} messages sent by
  wise processes after~\str{echo}.

  For the second case, let $\CG_{\text{max}}$ be the maximal guild.
  Consider the first wise process~$p_i$ in $\CG_{\text{max}}$ which sends
  \msg{ready}{m'} after \str{ready}.  From the protocol it follows that all
  processes in some kernel $K_i \in \CK_i$, which triggered $p_i$ to send
  \msg{ready}{m'}, have sent \msg{ready}{m'} to~$p_i$.  Moreover, according
  to the definition of a kernel, $K_i$ overlaps with all quorums for $p_i$.
  Since $p_i$ is in the (maximal) guild, at least one of the quorums for
  $p_i$ consists exclusively of wise processes.  Hence, some wise
  process~$p_j$ in the guild has sent \msg{ready}{m'} to $p_i$.  But since
  $p_i$ is the first wise process to send \str{ready} after \str{ready}, it
  follows that $p_j$ sent \msg{ready}{m'} after \str{echo}; therefore,
  $m' = m$ from the proof in the first case.  Continuing this argument
  inductively over all \str{ready} messages sent after \str{ready} by wise
  processes in $\CG_{\text{max}}$, in the order these were sent, shows that
  all those messages contain~$m$ and establishes the lemma.
\end{proof}

\begin{theorem}\label{thm:bracha}
  Algorithm~\ref{alg:bracha} implements asymmetric Byzantine reliable broadcast.
\end{theorem}

\begin{proof}
  Recall that the \emph{validity} property assumes there exists a maximal
  guild~$\CG_{\text{max}}$.  Since the sender $p_s$ is correct and
  according to asymmetric quorum availability, every process~$p_i$ in
  $\CG_{\text{max}}$ eventually receives a quorum of \msg{echo}{m} messages
  for itself, containing the message~$m$ from~$p_s$.  According to the
  protocol, $p_i$ therefore sends \msg{ready}{m} after \str{echo} unless
  $\var{sentready} = \true$; if this is the case, however, $p_i$ has
  already sent \msg{ready}{m} after \str{ready} as ensured by
  Lemma~\ref{lem:brachaunique}.  Hence, every process in $\CG_{\text{max}}$
  eventually sends \msg{ready}{m}.  Then every process $p_j$ in
  $\CG_{\text{max}}$ eventually receives a quorum for itself of
  \msg{ready}{m} messages and \op{r-delivers}~$m$, as ensured by the
  properties of a guild and by the protocol.
  
  To establish the \emph{totality} condition, suppose that some wise
  process~$p_i$ has \op{r-delivered} a message~$m$.  Then it has obtained
  \msg{ready}{m} messages from the processes in some quorum
  $Q_i \in \CQ_i$.
  Consider any other wise process~$p_j \in \CG_{\text{max}}$.  Since $p_i$
  and $p_j$ are both wise, it holds $F \in {\CF_i}^*$ and
  $F \in {\CF_j}^*$, which implies $F \in {\CF_i}^* \cap {\CF_j}^*$.  Then,
  the set $K = Q_i \setminus F$ intersects every quorum of $p_j$ by quorum
  consistency and therefore contains a kernel for $p_j$.  Since $K$
  consists only of correct processes, all of them have sent \msg{ready}{m}
  also to $p_j$ and $p_j$ eventually sends \msg{ready}{m} as well.  This
  implies that all wise processes in $\CG_{\text{max}}$ eventually send
  \msg{ready}{m} to all processes.  With the same argument as just given
  for validity, it follows that every wise process in the guild receives a
  quorum for itself of \msg{ready}{m} and \op{r-delivers}~$m$, as required
  for totality.

  The \emph{consistency} property follows immediately from the preceding
  argument and from Lemma~\ref{lem:brachaunique}, which implies that all
  wise processes deliver the same message.

  Finally, \emph{integrity} holds because of the \var{delivered} flag in
  the protocol and because of the argument showing validity together with
  Lemma~\ref{lem:brachaunique}.
\end{proof}

\begin{example}%
Consider again the protocol execution with $\BQ_{C}$ introduced earlier for
illustrating asymmetric consistent broadcast.  Recall that
$F = \{p_4^*, p_5^*\}$, the set $\{p_1, p_2, p_3\}$ is a guild, and $p_6$
is \naive.  The start of the execution is the same as shown previously and
omitted.  Instead of \op{c-delivering} $x$ and $u$, respectively, $p_1$ and
$p_6$ send \msg{ready}{x} and \msg{ready}{u} to all processes.
This is shown next, again with steps from left to right:
\begin{center}
  \small\vspace*{-4ex}
\begin{align*}
  & \dots  & p_1 &: \msg{ready}{x} \to \CP
     && && & p_1 &: \op{r-deliver}(x) \\
  & \dots  & p_2 &: \text{no quorum}        & p_2 &: \msg{ready}{x} \to \CP
        && & p_2 &: \op{r-deliver}(x) \\ 
  & \dots  & p_3 &: \text{no quorum}
        && & p_3 &: \msg{ready}{x} \to \CP  & p_3 &: \op{r-deliver}(x) \\ 
  & \dots  & p_4^* &: - \\
  & \dots  & p_5^* &: - \\
  & \dots  & p_6 &: \msg{ready}{u} \to \CP  & p_6 &: \text{no quorum}
\end{align*}
\end{center}
\vspace*{-1ex}

Note that the kernel systems of processes $p_1$, $p_2$, and $p_3$ are,
respectively,
$\CK_1 = \{\{p_1\}, \{p_3\}\}$, $\CK_2 = \{\{p_1\}, \{p_2\}\}$, and
$\CK_3 = \{\{p_2\}, \{p_3\}\}$.  Hence, when $p_2$ receives \msg{ready}{x}
from $p_1$, it sends \msg{ready}{x} in turn because $\{ p_1 \}$ is a kernel
for $p_2$, and when $p_3$ receives this message, then it sends
\msg{ready}{x} because $\{ p_2 \}$ is a kernel for $p_3$.

Furthermore, since $\{p_1, p_2, p_3\}$ is the maximal guild and contains a
quorum for each of its members, all three wise processes \op{r-deliver}~$x$
as implied by \emph{consistency} and \emph{totality}.  The \naive $p_6$
does not \op{r-deliver} anything, however.
\end{example}

\paragraph{Remarks.}
Asymmetric reliable broadcast (Definition~\ref{def:arbc}) ensures validity
and totality only for processes in the maximal guild.  There may exist wise
processes outside the maximal guild that do not terminate.  On the other
hand, asymmetric consistent broadcast (Definition~\ref{def:acbc}) ensures
validity also for all \emph{wise} processes.

Another open questions concerns the conditions for reacting to \str{ready}
messages in the asymmetric reliable broadcast protocol.  Already in
Bracha's protocol for the threshold
model~\cite{DBLP:journals/iandc/Bracha87}, a process (1) sends its own
\str{ready} message upon receiving $f+1$ \str{ready} messages and (2)
\op{r-delivers} an output upon receiving $2f+1$ \str{ready} messages.
These conditions generalize for arbitrary, non-threshold quorum systems to
receiving messages (1) from any set that is guaranteed to contain at least
one correct process and (2) from any set that still contains at least one
process even when any two fail-prone process sets are subtracted.  In
Algorithm~\ref{alg:bracha}, in contrast, a process delivers the payload
only after receiving \str{ready} messages from one of its quorums.  But
such a quorum (e.g., $\big\lceil\frac{n+f+1}{2}\big\rceil$ processes) may
be larger than a set in the second case (e.g., $2f+1$ processes).  It
remains interesting to find out whether this discrepancy is necessary.

\section{Consensus}
\label{sec:consensus}

In this section we define asymmetric asynchronous Byzantine consensus and
implement it through a randomized algorithm, which extends and improves the
protocol of Most{\'{e}}faoui \emph{et
  al.}~\cite{DBLP:conf/podc/MostefaouiMR14}.

The protocol of Most{\'{e}}faoui \emph{et al.} comes in multiple versions.
The original one, published at PODC
2014~\cite{DBLP:conf/podc/MostefaouiMR14} and where it also won the
best-paper award, suffers from a subtle and little-known liveness
problem~\cite{TG19}: an adversary can prevent progress among the correct
processes by controlling the messages between them and by sending them
values in a specific order.  The subsequent version (JACM
2015)~\cite{DBLP:journals/jacm/MostefaouiMR15} resolves this issue, but
requires many more communication steps and adds considerable complexity.

In Appendix~\ref{app:attack} we show in detail how it is possible to
violate liveness in the PODC 2014 version.  We also propose a method that
overcomes the problem, maintains the elegance of the protocol, and does not
affect its appealing properties. Based on this insight, in this section, we
show how to realize asynchronous consensus with asymmetric trust, again
with a protocol that maintains the simplicity of the original approach of
Most{\'{e}}faoui \emph{et al.}~\cite{DBLP:conf/podc/MostefaouiMR14}.

\subsection{Definition}

In an asynchronous binary consensus protocol, every correct process
initially \op{ac-proposes} a bit; the protocol concludes at a correct
process when it \op{ac-decides} a bit.  Our notion of Byzantine consensus
uses strong validity in the asymmetric model.  Furthermore, it restricts
the safety properties of consensus from all correct ones to \emph{wise}
processes in the guild.  For implementing asynchronous consensus, we use a
system enriched with randomization.  In round-based consensus algorithms,
the termination property is formulated with respect to the round number~$r$
that a process executes.  The corresponding probabilistic asymmetric termination
property is guaranteed only for wise processes in the maximal guild.

\begin{definition}[Asymmetric strong Byzantine consensus]\label{def:asymstrong}
  A protocol for asynchronous \emph{asymmetric strong Byzantine consensus}
  satisfies:

\begin{description}
\item[Probabilistic termination:] In all executions with a non-empty guild,
  every process in the maximal guild $\op{ac-decides}$ with probability
  $1$, i.e., for all $p_i \in \CG_{\text{max}}$,
  \[
    \lim_{r \rightarrow + \infty} (\P[\text{process $p_i$
      $\op{ac-decides}$ by round $r$}]) = 1.
  \]
  
\item[Strong validity:] In all executions with a non-empty guild, a wise
  process only $\op{ac-decides}$ a value that has been $\op{ac-proposed}$
  by some process in the maximal guild.
  
\item[Integrity:] No correct process $\op{ac-decides}$ twice.

\item[Agreement:] No two wise processes $\op{ac-decide}$ differently.

\end{description}
\end{definition}

The consensus protocol described here relies in a modular way on two
subprotocols.  Recall from Section~\ref{sec:model} that all processes are
connected pairwise by reliable FIFO links.  The FIFO guarantees on the
links hold across multiple protocol modules.

\subsection{Asymmetric common coin}

Our randomized consensus algorithm delegates its probabilistic choices to a
\emph{common coin} abstraction~\cite{DBLP:conf/focs/Rabin83,
  DBLP:books/daglib/0025983}.  This primitive is triggered by a
\op{release-coin} invocation and terminates by generating an
\op{output-coin}$(s)$ event, where $s \in \CB$ represents the random coin
value in a range~\CB.  We define this in the asymmetric-trust model.  The
coin remains hidden and unpredictable by faulty processes up to the time
when sufficiently many wise processes have released it.  This is the case
when at least a set of correct processes that is a kernel for all wise
processes have released it.

\begin{definition}[Asymmetric common coin]\label{def:acc}
  A protocol for \emph{asymmetric common coin} satisfies the
  following properties:
\begin{description}
\item[Termination:] In all executions with a non-empty guild, every process
  in the maximal guild eventually outputs a coin value.
  
\item[Unpredictability:] In all executions with a non-empty guild, no
  process has any information about the value of the coin before at least a
  kernel for all wise processes, which consists entirely of correct
  processes, has released the coin.

\item[Matching:] In all executions with a guild, with probability $1$ every process in the maximal guild outputs the same coin value.	

\item[No bias:] The distribution of the coin is uniform over $\mathcal{B}$.

\end{description}
\end{definition}

Here we consider binary consensus and $\CB = \{0,1\}$.  The
\emph{termination} property guarantees that every process in the maximal
guild eventually outputs a coin value that is ensured to be the same for
each of them by the \emph{matching} property.  The \emph{unpredictability}
property ensures that the coin value is kept secret in an execution until
at least a kernel for a wise process, consisting entirely of wise 
processes, releases the coin.  The existence of a kernel with only \emph{wise}
processes is required in order to avoid a liveness problem in the consensus
protocol (we describe this in Appendix~\ref{app:attack}).  The analogue of
this in the threshold symmetric model, where $f < n / 3$ processes may
fail, would be a coin with threshold $2f$, where the value is kept secret
until at least a set of $f+1$ \emph{correct} processes have released the
coin.  Finally, the \emph{no bias} property specifies the probability
distribution of the coin output.

\paragraph{The scheme.} 
We recall here the notion of the \emph{tolerated system} of an asymmetric
Byzantine quorum system from Section~\ref{sec:asymquorums}.  Every
asymmetric Byzantine quorum system~\BQ induces a tolerated system $\CT$
that contains sets $T$ that are the complement of the maximal guild in some
execution, i.e., $T = \CP \setminus \CG_{\max}$ and $\CG_{\max}$ is a
maximal guild for some execution and for~\BQ.  Crucial for our application
is the fact that \CT satisfies the $Q^3$-condition (Lemma~\ref{lem:b3-q3}),
hence one can construct a \emph{symmetric} Byzantine quorum system from
\CT.  In particular, the corresponding canonical system $\CH$ containing
all possible maximal guilds, is such a symmetric Byzantine quorum system.
The idea is to use the tolerated system as a ``bridge'' from the asymmetric
to the symmetric model, since reasoning is simpler in the latter.  At the
same time, this approach guarantees that in any execution where the system
is able to make progress because a non-empty guild exists, the protocol
can exploit the fact that such a guild exists also for a safety property.

The common coin scheme follows the approach of
Rabin~\cite{DBLP:conf/focs/Rabin83} and assumes that coins are
predistributed by a trusted dealer.  The scheme uses
Benaloh-Leichter~\cite{DBLP:conf/crypto/Leichter88} secret sharing, such
that the coin is additively shared within every maximal guild.  The dealer
shares one coin for every possible round of the protocol.  This requires
knowledge of the symmetric Byzantine quorum system~\CH corresponding to the
tolerated system~\CT.  Observe that every process can compute this because
\BF is globally known.

We assume that before the coin protocol runs, the dealer has chosen
uniformly at random a value $s \in \CB$ and shared it as follows.  For
every possible maximal guild $\CG = \{ p_{i_1}, \dots, p_{i_{m}} \}$ across all executions, the dealer
has picked uniform shares $s_{i_1}^\CG, \dots, s_{i_{m-1}}^\CG$ and set
$s_{i_m}^\CG = s + \sum_{\ell=1}^{m-1} s_{i_\ell}^\CG$.  Then the dealer has
given share $s_{i_\ell}^\CG$ to process $p_{i_\ell}$, for $\ell \in \{1,\dots,m\}$.
This implies that process $p_i$ holds a share for every guild of which it is a
member.

The code for process $p_i$ to release the coin is shown in
Algorithm~\ref{alg:acc}.  Specifically, when asked to release its coin
share (Lines~\ref{line:release-coin-begin}--\ref{line:release-coin-end},
Algorithm~\ref{alg:acc}), a process $p_i$ sends to all other processes a
share $s_\CG$ for each guild $\CG$ of which $p_i$ is a member.  Upon
receiving such shares, each process stores them in a local structure
(Lines~\ref{line:receive-share-begin}--\ref{line:receive-share-end},
Algorithm~\ref{alg:acc}).  When a process $p_i$ has enough shares, i.e.,
all shares from a guild~$\CG$, it can locally add them and output the coin
value (Lines~\ref{line:output-coin-begin}--\ref{line:output-coin-end},
Algorithm~\ref{alg:acc}).

\begin{algo*}[tbh]
\vbox{
\small
\begin{numbertabbing}\reset
  xxxx\=xxxx\=xxxx\=xxxx\=xxxx\=xxxx\=MMMMMMMMMMMMMMMMMMM\=\kill
  \textbf{State} \label{}\\
  \> $\CH$: set of all possible guilds \label{} \\
  \> $\var{share}[\CG][j]$: if $p_i \in \CG$, this holds the share received from $p_j$ for guild $\CG$; initially $\nil$ \label{} \\
  \\
  \textbf{upon event} \(\op{release-coin}\) \textbf{do} \label{line:release-coin-begin}\\
  \> \textbf{for all} \(\CG \in \CH\) \textbf{such that} $p_i \in \CG$ \textbf{do} \label{line:pi-in-G} \\
  \>\> let $s_{i\CG}$ be the share of $p_i$ for guild $\CG$ \label{} \\
  \>\> \textbf{for all} \(p_j \in \CP\) \textbf{do}  \label{} \\
  \>\>\> send message \msg{share}{s_{i\CG}, \CG, \var{round}} to \(p_j\) \label{line:release-coin-end} \\
  \\
  \textbf{upon} receiving a message \msg{share}{s, \CG, r} from \(p_j\)
  \textbf{such that} $r = \var{round} \textbf{ and } p_j \in \CG$ \textbf{do} \label{line:receive-share-begin} \\
  \> \textbf{if} \(\var{share}[\CG][j] = \bot\) \textbf{then}  \label{}\\
  \>\> \(\var{share}[\CG][j] \gets s\) \label{line:receive-share-end}\\
  \\
  \textbf{upon exists} \CG \textbf{such that}
    for all $j$ with $p_j \in \CG$, it holds \(\var{share}[\CG][j] \neq \nil \) \textbf{do} \label{line:output-coin-begin}\\
  \> $s \gets \sum_{j: p_j \in \CG} \var{share}[\CG][j]$ \label{line:compute-s}\\
  \> \textbf{output} $\op{output-coin(s)}$ \label{line:output-coin-end}\\[-5ex]
\end{numbertabbing}
}
\caption{Asymmetric common coin for round $\var{round}$ (code for~$p_i$)}
\label{alg:acc}
\end{algo*}

\begin{theorem}
  Algorithm~\ref{alg:acc} implements an asymmetric common coin.
\end{theorem}
\begin{proof}
  Let us consider an asymmetric fail-prone system $\mathbb{F}$ such that
  $B^3(\mathbb{F})$ holds and the corresponding asymmetric Byzantine quorum
  system $\mathbb{Q}$ for $\mathbb{F}$.  By Lemma~\ref{lem:b3-q3}, the
  tolerated system $\mathcal{T}$ of $\mathbb{Q}$ satisfies the
  $Q^3$-condition.  Let $\CH$ be the Byzantine quorum system for
  $\mathcal{T}$ consisting of all maximal guilds.  Assume an execution with
  a guild, where all processes in some $F \in \mathcal{T}^*$ are faulty and
  $\CG \in \CH$ is the maximal guild.
  
  For the \emph{termination} property, observe that every correct process,
  and hence also every process in $\CG$, invokes $\op{release-coin}$.  This
  implies that every process $p_i \in \CG$ sends \str{share} messages to
  all processes in $\CP$ (Line~\ref{line:release-coin-end},
  Algorithm~\ref{alg:acc}) containing the coin shares of $p_i$ for every
  guild in which $p_i$ belongs (Line~\ref{line:pi-in-G},
  Algorithm~\ref{alg:acc}), including $\CG$.  Eventually every correct
  process in $\CP$ receives a \str{share} message from every process in
  $\CG$, computes $s$ (Line~\ref{line:compute-s}, Algorithm~\ref{alg:acc})
  and triggers~$\op{output-coin}(s)$.  We note that termination holds
  actually for all correct processes, not just for those in the maximal
  guild.

  For the \emph{unpredictability} property, assume a correct process
  \( p_i \) outputs coin \( s \). This implies the existence of a set
  \( \mathcal{G}_k \in \mathcal{H} \), where each member of
  \( \mathcal{G}_k \) has sent a \str{share} message. Now, let us define
  \( K \) as the set \( \mathcal{G}_k \setminus F \). Observe that by
  construction, $K$ always contains a process $p_i$ in $\mathcal{G}$ that
  is wise, since \( \mathcal{G} \) is the maximal guild in the execution.
  This is a consequence of \( \mathcal{H} \) being a Byzantine quorum
  system, ensuring \( \mathcal{G}_k \cap \mathcal{G} \not\subseteq F
  \). This also implies that
  \( K = \mathcal{G}_k \setminus F \cap \mathcal{G} \neq \emptyset \).

  We first prove that this \( K \) intersects with every quorum of every wise
  process in the execution.
  Suppose by contradiction that there exists a wise process
  \( p_j \in \mathcal{P} \) with a quorum \( Q_j \in \mathcal{Q}_j \) such
  that \( Q_j \cap K = \emptyset \). Let \( p_i \) be a process in
  \( K \cap \mathcal{G} \).  Given \( K \subseteq \mathcal{G}_k \) and that
  \( \mathcal{G}_k \) is a guild within \( \mathcal{H} \), there must exist
  a quorum \( Q_i \in \mathcal{Q}_i \) for \( p_i \) such that
  \( Q_i \subseteq \mathcal{G}_k \). However, if
  \( Q_j \cap K = \emptyset \) and \( K = \mathcal{G}_k \setminus F \), it
  follows that \( Q_i \cap Q_j \subseteq F \). This situation contradicts
  the consistency property of the quorum system \( \mathbb{Q}
  \).

  Furthermore, employing the reasoning used in
  Lemma~\ref{lem:kernelinquorum}, we can derive from \( K \) a minimal set
  that continues to intersect with every quorum of every wise
  process. Therefore, it follows that \( K \) contains a kernel for every
  wise process, consisting only of correct processes.

  The \emph{matching} and \emph{no bias} properties follow directly from
  the fact that the coin value for every round is predetermined, albeit not
  known to any process, and chosen uniformly at random by the trusted
  dealer.
\end{proof}

\begin{example}
Let us consider a five-process asymmetric quorum system $\BQ_D$, defined through the following~$\BF_D$.

 \vspace*{-1ex}
\begin{minipage}[c]{0.05\linewidth}
  \vspace*{2ex}
  \center\Large$\BF_D$:
\end{minipage}
\begin{minipage}[c]{0.5\linewidth}
    \begin{eqnarray*}
    \mbox{}\\
    \CF_1 & = & \Theta^3_1(\{p_3,p_4,p_5\}) \\
    \CF_2 & = & \Theta^3_1(\{p_3,p_4,p_5\}) \\
    \CF_3 & = & \Theta^2_2(\{p_1,p_2\}) \vee \{p_4\} \vee \{p_5\} \\
    \CF_4 & = & \Theta^2_2(\{p_1,p_2\}) \vee \{p_3\} \vee \{p_5\}\\
    \CF_5 & = & \Theta^2_2(\{p_1,p_2\}) \vee \{p_3\} \vee \{p_4\} \\
    \mbox{}
  \end{eqnarray*}
 
  \end{minipage}
\begin{minipage}[c]{0.5\linewidth}
  \includegraphics[height=4cm]{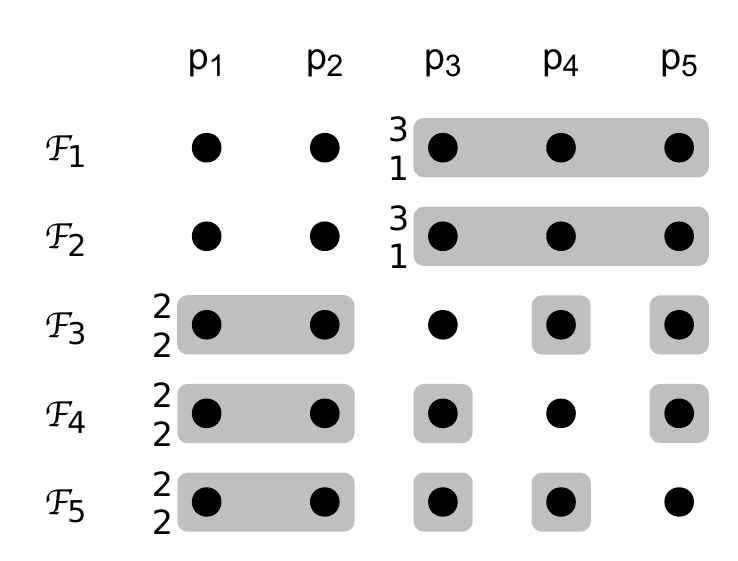}
\end{minipage}

The tolerated system is
$\CT = \{\{p_1, p_2\}, \{p_3\}, \{p_4\}, \{p_5\}\}$.  One can verify that
$B^3(\BF_D)$ holds; hence, by Lemma~\ref{lem:b3-q3}, also $Q^3(\CT)$ holds.
The corresponding symmetric Byzantine quorum system is
$\CH= \{\{p_3, p_4, p_5\}, \{p_1, p_2, p_4 , p_5\}, \{p_1, p_2, p_3 ,
p_5\}, \{p_1, p_2, p_3 , p_4\}\}$.  Observe that every $\CG \in \CH$ is a
guild in an execution in which the processes in $T = \CP \setminus \CG$ are
faulty.
  
Let us assume an execution with a set of faulty processes $F=\{p_1,p_2\}$;
this implies that the guild in this execution is $\{p_3,p_4,p_5\}$. We show how Algorithm~\ref{alg:acc} works.  

Let us assume that
the dealer has chosen $s=1$.  Then, for every guild $\CG_k\in \CH$, with
$k \in \{1,\ldots,4\}$, the dealer has chosen uniform shares as follows,
where $s^k_i$ denotes the share of process $i$ for guild $\CG_k$.
\begin{alignat*}{2}
  \CG_1 = \{p_3, p_4, p_5\}       &: s^1_3=1, s^1_4=0 \text{, and } s^1_5 = s + s^1_3 + s^1_4 = 0 \\
  \CG_2 = \{p_1, p_2, p_4 , p_5\} &: s^2_1=0, s^2_2=1, s^2_4=1 \text{, and } s^2_5 = s + s^2_1 + s^2_2 + s^2_4 = 1 \\
  \CG_3 = \{p_1, p_2, p_3 , p_5\} &: s^3_{1}=0, s^3_{2}=1, s^3_{3}=0 \text{, and } s^3_{5} = s + s^3_{1} + s^3_{2} + s^3_{3} = 0 \\
  \CG_4 = \{p_1, p_2, p_3 , p_4\} &: s^4_{1}=1, s^4_{2}=0, s^4_{3}=0 \text{, and } s^4_{4} = s + s^4_{1} + s^4_{2} + s^4_{3} = 0 
\end{alignat*}

Every process in $\CG_1=\{p_3,p_4,p_5\}$ upon \op{release-coin} sends a \str{share} message to every process $p_j \in \mathcal{P}$ for every share it has. 

Process $p_3$ is part of $\CG_1, \CG_3$ and $\CG_4$. This means that upon \op{release-coin}, $p_3$ sends \msg{share}{1, 1, 1}, \msg{share}{0, 3, 1} and \msg{share}{0, 4, 1} to every process in $\CP.$

Process $p_4$ is part of $\CG_1, \CG_2$ and $\CG_4$. This means that upon \op{release-coin}, $p_4$ sends \msg{share}{0, 1, 1}, \msg{share}{1, 2, 1} and \msg{share}{0, 4, 1} to every process in $\CP.$

Process $p_5$ is part of $\CG_1, \CG_2$ and $\CG_3$. This means that upon \op{release-coin}, $p_5$ sends \msg{share}{0, 1, 1}, \msg{share}{1, 2, 1} and \msg{share}{0, 3, 1} to every process in $\CP.$

Eventually every process in $\CG_1$ receives a \str{share} message
of the form \msg{share}{s^1_i, 1, 1} from each process $p_i \in \CG_1$,
computes $s \gets \sum_{i: p_i \in \CG_1} s^1_i$ (Line~\ref{line:compute-s}, Algorithm~\ref{alg:acc}) and $\op{output-coin}(1)$.
\end{example}

\paragraph{Discussion.}
This implementation is expensive because the number of shares for one
particular coin held by a process~$p_i$ is equal to the number of guilds in
which $p_i$ is contained.  It would be more efficient to implement an
asymmetric coin ``from scratch'' according to the protocols of Canetti and
Rabin \cite{DBLP:conf/stoc/CanettiR93} or of Patra \emph{et
  al.}~\cite{DBLP:journals/dc/PatraCR14}.  Alternatively, distributed
cryptographic implementations are possible, for example, implementations
relying on the hardness of the discrete logarithm
problem~\cite{DBLP:journals/joc/CachinKS05}.

\subsection{Asymmetric binary validated broadcast}

We generalize the binary validated broadcast as introduced by
Most{\'{e}}faoui \emph{et al.}~\cite{DBLP:conf/podc/MostefaouiMR14} and as
reviewed in~Appendix~\ref{app:bv-broadcast} to the asymmetric-trust model.
In this primitive, every process may broadcast a bit $b \in \{0,1\}$ by
invoking $\op{abv-broadcast}(b)$.  The primitive outputs at least one
binary value and possibly also both binary values through an
\op{abv-deliver} event.  This means one or two \op{abv-deliver} events
might occur at a correct process, which separates this notion from the
broadcasts of the previous section.
In the asymmetric version, all safety properties are restricted to wise
processes, and a guild is required for liveness.  This gives the following
notion.

\begin{definition}[Asymmetric binary validated broadcast]\label{def:abvb}
  A protocol for \emph{asymmetric binary validated broadcast} satisfies the
  following properties:
\begin{description}
\item[Validity:] In all executions with a guild, let $K$ be a kernel for every process in the maximal guild. If every process in $K$ is correct and
 has \op{abv-broadcast} the same value~$b \in \{0,1\}$, then every wise process eventually \op{abv-delivers}~$b$.
\item[Integrity:] In all executions with a guild, if a wise process \op{abv-delivers} some~$b$, then $b$ has been \op{abv-broadcast} by some
 process in the maximal guild. 
\item[Agreement:] In all executions with a guild, if a wise process
  \op{abv-delivers} some value~$b$, then every wise process
  eventually \op{abv-delivers}~$b$.
\item[Termination:] In all executions with a guild, every wise process eventually \op{abv-delivers} some value.
\end{description}
\end{definition}
 
Note that it guarantees properties only for processes that are wise.  Liveness properties also assume there exists a
guild.

\begin{algo*}[tbh]
\vbox{
\small
\begin{numbertabbing} \reset
  xxxx\=xxxx\=xxxx\=xxxx\=xxxx\=xxxx\=MMMMMMMMMMMMMMMMMMM\=\kill
  \textbf{State} \label{} \\
   \> \(\var{sentvalue} \gets [\str{false}]^2\): \(\var{sentvalue}[b]\) indicates whether $p_i$ has sent \msg{value}{b}  \label{}\\
  \>  \(\var{values} \gets [\emptyset]^n\): list of sets of received binary values \label{} \\
  \\
  \textbf{upon event} \(\op{abv-broadcast}(b)\) \textbf{do}  \label{}\\
  \> \(\var{sentvalue}[b] \gets \str{true}\)  \label{}\\
  \> send message \msg{value}{b} to all \(p_j \in \CP\)  \label{}\\
  \\
  \textbf{upon} receiving a message \msg{value}{b} from \(p_j\)
     \textbf{do}  \label{}\\
  \> \textbf{if} \(b \not\in \var{values}[j]\) \textbf{then}  \label{}\\
  \> \> \(\var{values}[j] \gets \var{values}[j] \cup \{b\}\)  \label{}\\
\\
  \textbf{upon exists} \(b \in \{0,1\} \) \textbf{such that}
     \( \{p_j \in \CP |~b \in ~\var{values}[j]\} \in \CK_i \)
     \textbf{and}  \(\neg \var{sentvalue}[b]\) \textbf{do}
     \` // a kernel for $p_i$    \label{algabvb:kernel}\\
  \> \(\var{sentvalue}[b] \gets \str{true}\)  \label{}\\
  \> send message \msg{value}{b} to all $p_j \in \CP$  \label{}\\
  \\
  \textbf{upon exists} \(b \in \{0,1\} \) \textbf{such that}
     \( \{p_j \in \CP |~b \in ~\var{values}[j] \} \in \CQ_i \) \textbf{do}
     \` // a quorum for $p_i$  \label{}\\
  \> \textbf{output} \(\op{abv-deliver(b)}\) \label{}\\[-5ex]
\end{numbertabbing}
}
\caption{Asymmetric binary validated broadcast (code for~$p_i$)}
\label{alg:abv-broadcast}
\end{algo*}

Algorithm~\ref{alg:abv-broadcast} works in the same way as the binary
validated broadcast by Most{\'{e}}faoui \emph{et
  al.}~\cite{DBLP:conf/podc/MostefaouiMR14}, but differs in the use of an
asymmetric quorum and kernel systems.  When a correct process $p_i$ invokes
$\op{abv-broadcast}(b)$ for $b \in \{0,1\}$, it sends a \str{value} message
containing $b$ to all processes.  Afterwards, whenever a correct process
$p_i$ receives \str{value} messages containing $b$ from from a kernel~$K_i$
for itself and has not sent a \str{value} message containing $b$ itself,
then it sends such message to every process.  Finally, once a correct
process $p_i$ receives \str{value} messages containing $b$ from a quorum
$Q_i$ for itself, it delivers $b$ through $\op{abv-deliver}(b)$.  Note that
a process \op{abv-delivers} at least one and at most two values.

\begin{theorem}\label{thm:abvbtheo}
  Algorithm~\ref{alg:abv-broadcast} implements asymmetric binary validated
  broadcast.
\end{theorem}

\begin{proof}
  To prove the \emph{validity} property, let us consider a kernel $K$ for
  every process $p_i$ in the maximal guild $\CG_{\text{max}}$. Moreover,
  let us assume that every process
  in $K$ has $\op{abv-broadcast}$ the same value $b \in \{0,1\}$. Then, by definition of a kernel, $K$
  intersects every $Q_i$ for every $p_i \in \CG_{\text{max}}$. According to
  the protocol, every process in $\CG_{\text{max}}$ eventually sends
  \msg{value}{b} unless $\var{sentvalue}[b] = \str{true}$ for some
  $p_i \in \CG_{\text{max}}$. However, if $\var{sentvalue}[b] = \str{true}$
  for~$p_i$, process~$p_i$ has already sent \msg{value}{b}. Since every process in the maximal guild 
  eventually sends \msg{value}{b}, eventually every correct process $p_j$ also receives \msg{value}{b} from a kernel for itself (see Corollary~\ref{cor:gmaxkernel}) and sends \msg{value}{b} unless $\var{sentvalue}[b] = \str{true}$. 
  However, as above, if $\var{sentvalue}[b] = \str{true}$
  for~$p_j$, process~$p_j$ has already sent \msg{value}{b}. It follows that eventually every wise
  process receives a quorum for itself of values~$b$
  and \op{abv-delivers}~$b$.

  For the \emph{integrity} property, let us assume an execution with a
  maximal guild $\CG_{\text{max}}$. Suppose first that only Byzantine
  processes \op{abv-broadcast}~$b$. Then, the set consisting of only these
  processes cannot form a kernel for any wise process. It follows that
  Line~\ref{algabvb:kernel} of Algorithm~\ref{alg:abv-broadcast} cannot be
  satisfied. If only \naive processes \op{abv-broadcast}~$b$, then by the
  definition of a quorum system and by the assumed existence of a maximal
  guild, there is at least one quorum for every process in
  $\CG_{\text{max}}$ that does not contain any \naive processes (e.g., as
  in Example~\ref{ex:wisegmax}). All \naive processes together cannot be a
  kernel for processes in $\CG_{\text{max}}$. Again,
  Line~\ref{algabvb:kernel} of Algorithm~\ref{alg:abv-broadcast} cannot be
  satisfied. Finally, let us assume that a wise process $p_i$ outside the
  maximal guild \op{abv-broadcasts}~$b$. Then, $p_i$ cannot be a kernel for
  every wise process: it is not part of the quorums inside
  $\CG_{\text{max}}$. It follows that if a wise process \op{abv-delivers}
  some~$b$, then $b$ has been \op{abv-broadcast} by some processes in the
  maximal guild.

  To show \emph{agreement}, let $F$ be the set of faulty processes and suppose that a wise process~$p_i$ has
  \op{abv-delivered}~$b$. Then it has obtained \msg{value}{b} messages
  from the processes in some quorum $Q_i \in \CQ_i$ and before from a
  kernel $K=Q_i \setminus F$ for itself. Each correct process in $K$ has sent
  \msg{value}{b} message to all other processes. Consider any other wise
  process~$p_j$.  Since $p_i$ and $p_j$ are both wise, we have
  $F \in {\CF_i}^*$ and $F \in {\CF_j}^*$, which implies
  $F \in {\CF_i}^* \cap {\CF_j}^*$. It follows that $K$ is also a kernel
  for $p_j$.  Thus, $p_j$ sends a \msg{value}{b} message to every
  process. This implies that all wise processes eventually send
  \msg{value}{b} to all processes. This also implies that eventually every process in $\CG_{\text{max}}$ sends
  \msg{value}{b}. By Corollary~\ref{cor:gmaxkernel}, $\CG_{\text{max}}$ contains a kernel for every correct process $p_k$. Thus, $p_k$ sends a \msg{value}{b} message to every
  process. 
  Therefore eventually every wise process receives a quorum for itself of \msg{value}{b} messages and
  \op{abv-deliver}~$b$.

  For the \emph{termination} property, let us assume an execution with a
  maximal guild $\CG_{\text{max}}$ and set of faulty processes $F$. Note
  that in any execution, every process in $\mathcal{P} \setminus F$
  \op{abv-broadcasts} some binary values.  We show that there is a set
  $K \subseteq \mathcal{P} \setminus F$ such that $K$ is a kernel for every
  process in the maximal guild consisting of correct processes and every process in $K$ \op{abv-broadcasts}
  the same value $b \in \{0,1\}$.  Observe that a correct process initially
  \op{abv-broadcasts} only one value in $\{0,1\}$. So, let
  $\mathcal{P} \setminus F = S_0 \cup S_1$ with $S_0$ and $S_1$ two sets of processes such that
  $S_0 \cap S_1 = \emptyset$ and such that every process in $S_0$
  \op{abv-broadcasts}~$b$ and every process in $S_1$
  \op{abv-broadcasts}~$1-b = \overline{b}$. Moreover, let us assume that
  neither $S_0$ nor $S_1$ contains a kernel for every process in the
  maximal guild.  If $S_0$ does not contain a kernel for a process in the
  maximal guild, then there exists a process
  $p_j \in \CG_{\text{max}}$ and a quorum $Q_j$ for $p_j$ such that
  $Q_j \cap S_0 = \emptyset$. This means that every correct process in
  $Q_j$ \op{abv-broadcasts}~$\overline{b}$. Similarly, if $S_1$ does not
  contain a kernel for a process in the maximal guild, then there
  exists a process $p_k \in \CG_{\text{max}}$ and a quorum $Q_k$ for $p_k$
  such that $Q_k \cap S_1 = \emptyset.$ This means that every correct
  process in $Q_k$ \op{abv-broadcasts} $b$.  However, if this is the case,
  then $Q_j \cap Q_k \subseteq F$, which contradicts the consistency
  property of an asymmetric Byzantine quorum system, given that $p_j$ and
  $p_k$ are both wise.  This implies that either $S_0$ or $S_1$ contains a
  kernel $K$ for every process in the maximal guild consisting of correct
  processes and such that every process in $K$ \op{abv-broadcasts} the same
  value.  Termination then follows from the validity property.
\end{proof}

\subsection{Asymmetric randomized consensus}

In consensus, a correct process may \emph{propose} a binary value $b$ by
invoking $\op{ac-propose}(b)$, and the consensus abstraction \emph{decides}
for $b$ through an $\op{ac-decide}(b)$ event.

\newcommand\est{\var{est}\xspace}
\begin{algo*}
\vbox{
\small
\begin{numbertabbing}\reset
  xxxx\=xxxx\=xxxx\=xxxx\=xxxx\=xxxx\=MMMMMMMMMMMMMMMMMMM\=\kill
  \textbf{State}  \label{}\\
  \> $\var{round} \gets 0$: current round  \label{}\\
  \> $\var{values} \gets \{\}$: set of \op{abv-delivered} binary values for
     the round  \label{}\\
  \> $\var{aux} \gets [\{\}]^n$: stores sets of values that have
     been received in \str{aux} messages in the round \label{} \\
    \> $\var{decided} \gets []^n$: stores binary values that have
     been reported as decided by other processes   \label{}\\
     \> $\var{sentdecide} \gets \false$: indicates whether $p_i$ has sent a \str{decide} message  \label{}\\
 \\
 \textbf{upon event} \(\op{ac-propose}(b)\) \textbf{do}   \label{}\\
 \> \textbf{invoke} $\op{abv-broadcast}(b)$ with tag~\var{round} \label{} \\
\\
 \textbf{upon} $\op{abv-deliver}(b)$ with tag $r$ \textbf{such that}
    $r = \var{round}$ \textbf{do}  \label{}\\
 \> $\var{values} \gets \var{values} \cup \{b\}$  \label{}\\
 \> send message \msg{aux}{\var{round}, b} to all $p_j \in \CP$  \label{}\\
  \\
 \textbf{upon} receiving a message \msg{aux}{r, b} from $p_j$
    \textbf{such that} $r = \var{round}$ \textbf{do}  \label{}\\
 \> $\var{aux}[j] \gets \var{aux}[j] \cup \{b\}$  \label{}\\
\\
 \textbf{upon exist}
 \(
   \{p_j \in \CP \,|\, \var{aux}[j] \subseteq \var{values}\} \in \CQ_i \) \textbf{do}
 \` // a quorum for $p_i$  \label{}\\
 \> $\op{release-coin}$ with tag~\var{round}  \label{}\\
 \\
 \textbf{upon event} $\op{output-coin}(s)$ with tag~$\var{round}$ \textbf{and} \label{algarbc:coin1} \\
  \>\>\textbf{exists}~$B \subseteq \{0,1\}$, $Q_i \in \mathcal{Q}_i$ \textbf{such that} $B \neq \emptyset$ \textbf{and} for all $p_j \in Q_i$ it holds $B = \var{aux}[j]$ \textbf{do}  \label{algarbc:coin2} \\
 \> \(\var{round} \gets \var{round} + 1\)  \label{}\\
 \> \textbf{if exists} $b$ \textbf{such that} $|B| = 1 \land B = \{b\}$ \textbf{then} \label{} \\
 \>\> \textbf{if} $b = s \land \neg \var{sentdecide}$ \textbf{then}  \label{}\\
 \>\>\> send message \msg{decide}{b} to all $p_j \in \CP$ \label{} \\
      \>\>\> $\var{sentdecide} \gets \true$  \label{}\\
      \>\> \textbf{invoke} $\op{abv-broadcast}(b)$ with tag~\var{round}
      \` // propose $b$ for the next round \label{algarbc:proposeb} \\
 \> \textbf{else} \label{} \\
 \> \> \textbf{invoke} $\op{abv-broadcast}(s)$ with tag~\var{round}
      \` // propose coin value $s$ for the next round \label{algarbc:proposes} \\
 \> \(\var{values} \gets [\perp]^n\)  \label{}\\
 \> $\var{aux} \gets [\{\}]^n$ \label{}\\
 \\
   \textbf{upon} receiving a message \msg{decide}{b} from $p_j$ \textbf{such that} $\var{decided}[j] = \bot$ \textbf{do}  \label{}\\
 \> $\var{decided}[j] = b$  \label{}\\
\\
 \textbf{upon exists} \(b \neq \bot\) \textbf{such that} \(\{p_j \in \CP \,|\, \var{decided}[j] = b\} \in \CK_i\) \textbf{do}
 \` // a kernel for $p_i$  \label{}\\
 \> \textbf{if} $\neg\var{sentdecide}$ \textbf{then}  \label{}\\
 \>\> send message \msg{decide}{b} to all $p_j \in \CP$  \label{}\\
 \>\> $\var{sentdecide} \gets \true$ \label{}\\
  \\
  \textbf{upon exists} \(b \neq \bot\) \textbf{such that} \(\{p_j \in \CP \,|\, \var{decided}[j] = b\} \in \CQ_i\) \textbf{do}
 \` // a quorum for $p_i$  \label{}\\
 \> \(\op{ac-decide}(b)\)  \label{} \\
 \> \textbf{halt} \label{}\\[-5ex]
\end{numbertabbing}
}
\caption{Asymmetric randomized binary consensus (code for $p_i$).}
\label{alg:arbv-consensus}
\end{algo*}

Similar to the protocol of Most{\'{e}}faoui \emph{et
  al.}~\cite{DBLP:journals/jacm/MostefaouiMR15},
Algorithm~\ref{alg:arbv-consensus} proceeds in rounds, and in each round an
instance of $\op{abv-broadcast}$ is invoked. A correct process $p_i$
executes $\op{abv-broadcast}$ and waits for a value $b$ identified by a tag
characterizing the current round. Once received, $p_i$ adds $b$ to
$\var{values}$, broadcasts $b$ in an \str{aux} message to all other
processes, and all of them will eventually add $b$ to~$\var{aux}$.
The \str{aux} messages serve to ``enhance'' the distributed knowledge about
the valid decision values, which must have been \op{abv-proposed}
by processes in the guild.  When $p_i$ has received
a set $B \subseteq \var{values}$ of values carried by \str{aux} messages
from all processes in a quorum $Q_i$ for itself, then $p_i$ releases its
coin with tag~$r$.  Process~$p_i$ then waits for $\op{output-coin}$ with
tag $r$ and the common coin value~$s$. Observe that
Algorithm~\ref{alg:arbv-consensus} allows the set $B$ to change while
reconstructing the common coin
(Lines~\ref{algarbc:coin1}--\ref{algarbc:coin2}).

Subsequently, $p_i$ checks if there is a single value $b$ in $B$.  If so,
and if $b=s$, then $p_i$ becomes ready to decide~$b$ and it does so by
broadcasting a \str{decide} message with value~$b$ to every process.  If
there is more than one value in $B$, then $p_i$ changes its proposal to
$s$.  In any case, the process starts another round and invokes a new
instance of $\op{abv-broadcast}$ with its proposal.

In parallel, the protocol potentially disseminates \str{decide} messages
and may terminate.  When $p_i$ receives a \str{decide} message from a kernel
of processes for itself containing the same value~$b$, then it broadcasts a
\str{decide} message itself containing $b$ to every process, unless it has
already done so.  Once $p_i$ has received a \str{decide} message from a
quorum of processes for itself with the same value~$b$, it
$\op{ac-decides}(b)$ and halts.  This ``amplification'' step is reminiscent
of Bracha's reliable broadcast
protocol~\cite{DBLP:journals/iandc/Bracha87}.  Hence, the protocol does not
execute rounds forever, in contrast to the original formulation of
Most{\'{e}}faoui \emph{et al.}~\cite{DBLP:journals/jacm/MostefaouiMR15},
which satisfies a weaker notion of termination.

The following lemma illustrates that the problem described by Tholoniat and
Gramoli~\cite{TG19} and described in Appendix~\ref{app:attack} does not
occur in this protocol.  This lemma is not directly used in the analysis of
Algorithm~\ref{alg:arbv-consensus}.

\begin{lemma}\label{lem:attack}
  If a wise process $p_i$ executes $\op{output-coin}(s)$ and has
  $B = \{0,1\}$, then every other wise process that $\op{output-coin}(s)$
  has also $B=\{0,1\}$.
\end{lemma}

\begin{proof}
  Let us assume that a wise process~$p_i$ executes $\op{output-coin}(s)$
  while it stores $B = \{0,1\}$.  By inspection of the common-coin
  implementation, this means that $p_i$ has received \str{share} messages
  from every process in some guild $\mathcal{G}$
  (Line~\ref{line:output-coin-begin}, Algorithm~\ref{alg:acc}) and has
  $B=\var{aux}[j]=\{0,1\}$ for all $p_j$ in a quorum $Q_i$ for
  $p_i$. Observe that because $p_i$ is wise, $Q_i \cap \CG$ contains some
  correct process.

  Consider another wise process $p_j$ that has also obtained
  $\op{output-coin}(s)$. It follows that $p_j$ has received \str{share}
  messages from some guild $\mathcal{G}'$ as well. Observe that $p_i$ and
  $p_j$, before receiving the \str{share} messages from every process in
  $\mathcal{G}$ and $\mathcal{G}'$, respectively, receive all \str{aux}
  messages that the correct processes in these guilds have sent before the
  \str{share} messages. This follows from the assumption of FIFO reliable
  point-to-point links across the protocols.

  Recall from Lemma~\ref{lem:b3-q3} that the set of guilds is a symmetric
  Byzantine quorum system for the tolerated system $\CT$ of $\BQ$.  Quorum
  consistency then implies that $\mathcal{G}$ and $\mathcal{G}'$ have some
  correct process(es) in common.  So, according to the reasoning above,
  $p_i$ and $p_j$ receive some \str{aux} messages from the same correct
  process before they may output the coin.  This means that if $p_i$ has
  $B=\{0,1\}$ after $\op{output-coin}(s)$, then every quorum $Q_j$ for
  $p_j$ will contain a process $p_k$ such that $\var{aux}[k]=\{0,1\}$ for
  $p_j$.  Every wise process therefore must eventually have $B=\{0,1\}$.
\end{proof}

\begin{theorem}\label{thm:arbvtheo}
  Algorithm~\ref{alg:arbv-consensus} implements asymmetric strong Byzantine consensus.
\end{theorem}

\begin{proof}
  To prove the \emph{strong validity} property, let us assume that a wise
  process $p_i$ has $\op{ac-decided}$ a value~$b$. This means that $p_i$
  has received \msg{decide}{b} messages from a quorum $Q_i$ for itself.
  Moreover, before deciding, process $p_i$ has received \msg{decide}{b}
  messages from a kernel~$\mathcal{K}_i$ for itself and sent
  \msg{decide}{b} to every other process.

  Whenever a correct process~$p_i$ has sent such a \str{decide} message
  containing $b$ in a round~$r$, it has obtained $B = \{b\}$ and $b$ is the
  same as the coin value in the round. Then, $p_i$ has received $b$ from a
  quorum $Q_i$ for itself through \str{aux} messages. Every process in
  $Q_i$ has received a \msg{aux}{r,b} message and $b$ has been
  \op{abv-delivered}.  According to the integrity property of the validated
  broadcast, $b$ has been \op{abv-broadcast} by a process in the maximal
  guild and, specifically, $\var{values}$ contains only values
  \op{abv-broadcast} by processes in the maximal guild. It follows that $b$
  has been proposed by some processes in the maximal guild.

  For the \emph{agreement} property, suppose that a wise process has
  received \msg{aux}{r,b} messages from a quorum $Q_i$ for itself. Consider
  any other wise process $p_j$ that has received a quorum $Q_j$ for itself
  of \msg{aux}{r,\overline{b}} messages. If at the end of round $r$ there
  is only one value in $B$, then from consistency property of quorum
  systems, it follows $b = \overline{b}$. Furthermore, if $b=s$ then $p_i$
  and $p_j$ broadcast a \msg{decide}{b} message to every process and decide
  for $b$ after receiving a quorum of \msg{decide}{b} messages for
  themselves, otherwise they both $\op{abv-broadcast}(b)$ and they continue
  to $\op{abv-broadcast}(b)$ until $b=s$. If $B$ contains more than one
  value, then $p_i$ and $p_j$ proceed to the next round and invoke a new
  instance of $\op{abv-broadcast}$ with~$s$. Therefore, at the beginning of
  the next round, the proposed values of all wise processes are equal. The
  property easily follows.

  For the \emph{integrity} property, notice that the process halts after
  \op{ac-deciding} and therefore does not \op{ac-decide} more than once.

  The \emph{probabilistic termination} property follows from two
  observations. First, the termination and the agreement properties of
  binary validated broadcast imply that every wise process 
  \op{abv-delivers} the same binary value from the validated broadcast
  instance and this value has been \op{abv-broadcast} by some processes in
  the maximal guild. Second, we show that with probability $1$, there
  exists a round at the end of which all processes in $\CG_{\text{max}}$
  have the same proposal~$b$. If at the end of round~$r$, every process in
  $\CG_{\text{max}}$ has proposed the coin value
  (Line~\ref{algarbc:proposes}, Algorithm~\ref{alg:arbv-consensus}), then
  all of them start the next round with the same value. Similarly, if every
  process in $\CG_{\text{max}}$ has executed Line~\ref{algarbc:proposeb}
  (Algorithm~\ref{alg:arbv-consensus}) they adopt the value $b$ and start
  the next round with the same value.

  However, it could be the case that some wise process in the maximal guild
  \op{abv-broadcasts} a bit $b$ in Line~\ref{algarbc:proposeb} and another
  such process \op{abv-broadcasts} the coin output~$s$ in
  Line~\ref{algarbc:proposes}.  Observe that the properties of the common
  coin abstraction guarantee that the coin value is random and chosen
  independently of~$b$.
  In particular, the \emph{unpredictability} of the common coin ensures
  that no information about~$s$ is revealed until some kernel~$K$ for all
  wise processes, which consists only of correct processes, has released
  the coin.  But for every wise process~$p_i$, this kernel $K$ will
  intersect a quorum $Q_i$ in the condition of
  Line~\ref{algarbc:coin2}. Since the \str{aux} messages from the processes
  $Q_i$ determine $B = \{b\}$ for every wise process that
  \op{abv-broadcasts}~$b$, all processes in $K \cap Q_i$ must have received
  the same value~$b$ before information about the coin can become public.
  Hence $b$ is independent of the random value~$s$ and they with
  probability~$\frac{1}{2}$. The probability that there exists a round $r'$
  in which the coin equals the value $b$ proposed by all processes in
  $\CG_{\text{max}}$ during round $r'$ approaches $1$ when $r$ goes to
  infinity.

  Let $r$ thus be some round in which every process in $\CG_{\text{max}}$
  \op{abv-broadcasts} the same value~$b$; then, none of them will ever change
  their proposal again. This is due to the fact that every wise process
  invokes an binary validated broadcast instance with the same proposal~$b$.
  According to the validity and agreement properties of asymmetric
  binary validated broadcast, every wise process then
  \op{bv-delivers} the same, unique value~$b$. Hence, the proposal of every wise
  process is set to $b$ and does not change in future
  rounds. Finally, the properties of common coin guarantee that the
  processes eventually reach a round in which the coin outputs~$b$.
  Therefore, with probability $1$ every process in the maximal guild sends
  a \str{decide} message with value~$b$ to every process in that round.
  This implies that it exists a quorum $Q_i \subseteq \CG_{\text{max}}$ for
  a process $p_i \in \CG_{\text{max}}$ such that every process in $Q_i$ has
  sent a \str{decide} message with value~$b$ to every process.
  Moreover, the set of processes in the maximal guild contains a kernel for $p_i$ and for every other correct process $p_j$ (Corollary~\ref{cor:gmaxkernel}). If a correct process $p_j$ receives a \str{decide} message with value~$b$ from a kernel for itself, it sends a \str{decide} message with value~$b$ to every process unless it has already done so. 
  It follows that eventually every wise process receives \str{decide}
  messages with the value~$b$ from a quorum for itself and \op{ac-decides}
  for~$b$.
\end{proof}

\section{Conclusion}

This work has introduced asymmetric Byzantine quorum systems, which enable
distributed fault-tolerant protocols with subjective trust assumptions.
The asymmetric-trust model is a strict generalization of Byzantine quorum
systems and intended to work with generic extensions of the standard
protocols, where Byzantine quorums are used.  Indeed, this paper has shown
how register emulations, Byzantine consistent and reliable broadcasts, and
randomized asynchronous consensus can be extended to asymmetric trust.
Some of existing protocols had to be changed in subtle ways because not
only asymmetric quorums play a role but also further concepts, such as core
sets and kernels.  This work has also extended these notions to asymmetric
trust.

The changes to existing protocols follow a general pattern.  The most
important one is that when a process~$p_i$ obtains a number of responses
from a (Byzantine) quorum, which consists in the threshold case of any set
with more than $\lceil \frac{n+f+1}{2} \rceil < n-f$ processes, this is
replaced by the step of $p_i$ receiving responses from one of its
quorums~$Q_i$.  Waiting for a core set of responses, which means $f+1$
messages in the threshold case, changes to obtaining a core set~$C_i$
for~$p_i$ of responses or a kernel~$K_i$ for~$p_i$ of responses,
respectively.  The appropriate notion depends on the context.

There exist a considerable number of more elaborate distributed protocols
in the Byzantine-fault model, notably for consensus and total-order
broadcast.  It is expected that these can be generalized as well to
asymmetric quorums, but the actual formulations remain open.  Furthermore,
many Byzantine-tolerant distributed protocols rely on distributed
cryptographic primitives.  It is an interesting problem to generalize them
to subjective trust assumptions in a scalable and efficient way.

\appendix

\section*{Appendix}

\section{Revisiting signature-free asynchronous Byzantine consensus}
\label{app:attack}

In 2014, Most{\'{e}}faoui \emph{et
  al.}~\cite{DBLP:conf/podc/MostefaouiMR14} introduced a round-based
asynchronous randomized consensus algorithm for binary values. It had
received considerable attention because it was the first protocol with
optimal resilience, tolerating up to $f < \frac{n}{3}$ Byzantine processes,
that did not use digital signatures.  Hence, this protocol needs only
authenticated channels and remains secure against a computationally
unbounded adversary.  Moreover, it takes $O(n^2)$ constant-sized messages
in expectation and has a particularly simple structure.  Our description
here excludes the necessary cost for implementing randomization, for which
the protocol relies on an abstract common coin primitive, as defined by
Rabin~\cite{DBLP:conf/focs/Rabin83}.

This protocol, which we call the \emph{PODC-14}
version~\cite{DBLP:conf/podc/MostefaouiMR14} in the following, suffers from
a subtle and little-known problem.  It may violate liveness, as has been
explicitly mentioned by Tholoniat and Gramoli~\cite{TG19}.  The
corresponding journal publication by Most{\'{e}}faoui \emph{et
  al.}~\cite{DBLP:journals/jacm/MostefaouiMR15}, to which we refer as the
\emph{JACM-15} version, touches briefly on the issue and goes on to present
an extended protocol.  This fixes the problem, but requires also many more
communication steps and adds considerable complexity.

The purpose of this appendix is to revisit the PODC-14 protocol, to point
out in detail how the protocol may fail, and to introduce a compact
solution for fixing it, all in a self-contained way.  For the same reason,
we use the symmetric threshold-fault model here, where any $f$ out of $n$
processes may be faulty.

We discovered discovered this solution while extending the randomized
consensus algorithm to asymmetric quorums.  The corresponding asymmetric
randomized Byzantine consensus protocol appeared in
Section~\ref{sec:consensus} and is proven secure there.

Before addressing randomized consensus, we recall the key abstraction introduced in the PODC-14 paper, a protocol for broadcasting binary values.

\subsection{Binary-value broadcast}
\label{app:bv-broadcast}

The \emph{binary validated broadcast} primitive has been introduced in the PODC-14 version~\cite{DBLP:conf/podc/MostefaouiMR14} under the name \emph{binary-value broadcast}.\footnote{Compared to their work, we adjusted some conditions to standard terminology and chose to call the primitive ``binary \emph{validated} broadcast'' to better emphasize its aspect of validating that a delivered value was broadcast by a correct process.}
In this primitive, every process may broadcast a bit $b \in \{0,1\}$ by invoking $\op{bv-broadcast}(b)$. The broadcast primitive outputs at least one  value $b$ and possibly also both binary values through a $\op{bv-deliver}(b)$ event, according to the following notion. 

\begin{definition}[Binary validated broadcast]\label{def:bvb}
  A protocol for \emph{binary validated broadcast} satisfies the
  following properties:
\begin{description}
\item[Validity:] If at least $(f+1)$ correct processes \op{bv-broadcast} the same value~$b \in \{0,1\}$, then every correct process eventually \op{bv-delivers}~$b$.
\item[Integrity:] A correct process \op{bv-delivers} a particular value~$b$ at most once and only if $b$ has been \op{bv-broadcast} by some correct process.
\item[Agreement:] If a correct process \op{bv-delivers} some value $b$, then every correct process eventually \op{bv-delivers}~$b$.
\item[Termination:] Every correct process eventually \op{bv-delivers} some value $b$.
\end{description}
\end{definition}

The implementation given by Most{\'{e}}faoui \emph{et
  al.}~\cite{DBLP:conf/podc/MostefaouiMR14} works as follows. When a
correct process~$p_i$ invokes $\op{bv-broadcast}(b)$ for $b \in \{0,1\}$,
it sends a \str{value} message containing $b$ to all processes.
Afterwards, whenever a correct process receives \str{value} messages
containing $b$ from at least $f+1$ processes and has not itself sent a
\str{value} message containing $b$, then it sends such message to every
process.  Finally, once a correct process receives \str{value} messages
containing $b$ from at least $2f+1$ processes, it delivers $b$ through
$\op{bv-deliver}(b)$.  Note that a process may \op{bv-deliver} up to two
values.  A formal description, in the asymmetric model, appeared in
Algorithm~\ref{alg:abv-broadcast} (Section~\ref{sec:consensus}).

\subsection{Randomized consensus}

We recall the notion of \emph{randomized Byzantine consensus} here and its
implementation by Most{\'{e}}faoui \emph{et
  al.}~\cite{DBLP:conf/podc/MostefaouiMR14}.
In a consensus primitive, every correct process proposes a value~$v$ by
invoking \(\op{propose}(v)\), which typically triggers the start of the
protocol among processes; it obtains as output a decided value $v$ through
a \(\op{decide}(v)\) event. There are no assumptions made about the faulty
processes.
We use the probabilistic termination property for round-based protocols.
It requires that the probability that a correct process decides after
executing infinitely many rounds approaches~1.

\begin{definition}[Strong Byzantine consensus]\label{def:strong}
  A protocol for asynchronous \emph{strong Byzantine consensus} satisfies:

\begin{description}
\item[Probabilistic termination:] Every correct process~$p_i$ decides with probability $1$, in the sense
  that
  \[
    \lim_{r \rightarrow + \infty}
      \P[\text{a correct process $p_i$ decides by round $r$}] = 1.
  \]
  
\item[Strong validity:] A correct process only decides a value that has been proposed by some correct process.
  
\item[Integrity:] No correct process decides twice.

\item[Agreement:] No two correct processes decide differently.

\end{description}

\end{definition}

The probabilistic termination and integrity properties together imply that
every correct process decides exactly once, while the agreement property
ensures that the decided values are equal. Strong validity asks that if all
correct processes propose the same value~$v$, then no correct process
decides a value different from~$v$. Otherwise, a correct process may only
decide a value that was proposed by some correct
process~\cite{DBLP:books/daglib/0025983}.  In a \emph{binary} consensus
protocol, as considered here, only 0 and 1 may be proposed.

The implementation of randomized consensus by Most{\'{e}}faoui \emph{et
  al.}~\cite{DBLP:conf/podc/MostefaouiMR14} delegates its probabilistic
choices to a \emph{common coin} abstraction \cite{DBLP:conf/focs/Rabin83,
  DBLP:books/daglib/0025983}, a random source observable by all processes
but unpredictable for an adversary. A common coin is invoked at every
process by triggering a \op{release-coin} event. We say that a process
\emph{releases} a coin because its value is unpredictable, unless more than
$f$ correct processes have invoked the coin. The value $s \in \CB$ of the
coin with tag $r$ is output through an event \op{output-coin}.

\begin{definition}[Common coin]\label{def:cc}
  A protocol for \emph{common coin} satisfies the following properties:
\begin{description}
\item[Termination:] Every correct process eventually outputs a coin value. 
  
\item[Unpredictability:] Unless more than $2f$
  processes have released the coin, no process has any information about
  the coin output by a correct process.
  
\item[Matching:] With probability 1 every correct process outputs the same coin value.	

\item[No bias:] The distribution of the coin is uniform over $\mathcal{B}$.

\end{description}
\end{definition}

Observe that the unpredictability condition implies that at least $f+1$
\emph{correct} processes are required to release the coin in order for a
process to have information about the coin value output by a correct
process.


\begin{algo*}
\vbox{
\small
\begin{numbertabbing}\reset

  xxxx\=xxxx\=xxxx\=xxxx\=xxxx\=xxxx\=MMMMMMMMMMMMMMMMMMM\=\kill
  \textbf{State} \label{} \\
  \> $\var{round} \gets 0$: current round \label{} \\
  \> $\var{values} \gets \{\}$: set of \op{bv-delivered} binary values for
     the round  \label{}\\
  \> $\var{aux} \gets [\{\}]^n$: stores sets of values that have
     been received in \str{aux} messages in the round  \label{}\\
  \\
 \textbf{upon event} \(\op{rbc-propose}(b)\) \textbf{do}  \label{} \\
 \> \textbf{invoke} $\op{bv-broadcast}(b)$ with tag~\var{round}  \label{}\\
\\
 \textbf{upon} $\op{bv-deliver}(b)$ with tag $r$ \textbf{such that}
    $r = \var{round}$ \textbf{do} \label{} \\
 \> $\var{values} \gets \var{values} \cup \{b\}$  \label{}\\
 \> send message \msg{aux}{\var{round}, b} to all $p_j \in \CP$ \label{} \\
\\
 \textbf{upon} receiving a message \msg{aux}{r, b} from $p_j$
    \textbf{such that} $r = \var{round}$ \textbf{do} \label{} \\
 \> $\var{aux}[j] \gets \var{aux}[j] \cup \{b\}$  \label{}\\
  \\
 \textbf{upon exists} $B \subseteq \var{values}$ \textbf{such that} 
    \( B \neq \{\} \) \textbf{and} \(
    |\{p_j \in \CP \,|\, B = \var{aux}[j]\}| \ge n-f \) \textbf{do}
    \label{algrbc:uponb} \\
 \> $\op{release-coin}$ with tag~\var{round}  \label{}\\
 \> \textbf{wait for} $\op{output-coin}(s)$ with tag~$\var{round}$  \label{}\\
 \> \(\var{round} \gets \var{round} + 1\)  \label{}\\
 \> \textbf{if exists} $b$ \textbf{such that} $B = \{b\}$ \textbf{then}
    \` // i.e., $|B| = 1$ \label{algrbc:ifb} \\
 \> \> \textbf{if} $b = s$ \textbf{then} \label{} \\
 \>\>\> \textbf{output} \(\op{rbc-decide(b)}\) \label{}\\
 \>\> \textbf{invoke} $\op{bv-broadcast}(b)$ with tag~\var{round}
      \` // propose $b$ for the next round  \label{algrbc:bvbit}\\
 \> \textbf{else} \label{} \\
 \> \> \textbf{invoke}  $\op{bv-broadcast}(s)$ with tag~\var{round}
      \` // propose coin value $s$ for the next round  \label{algrbc:bvcoin}\\
 \> \(\var{values} \gets [\perp]^n\)  \label{}\\
 \> $\var{aux} \gets [\{\}]^n$  \label{}\\[-5ex]
\end{numbertabbing}
}
\caption{Randomized binary consensus according to Most{\'{e}}faoui \emph{et al.}~\cite{DBLP:conf/podc/MostefaouiMR14}  (code for $p_i$).}
\label{alg:rbv-consensus}
\end{algo*}

We now recall the implementation of strong Byzantine consensus according to
Most{\'{e}}faoui \emph{et al.}~\cite{DBLP:conf/podc/MostefaouiMR14} in the
PODC-14 version, shown in Algorithm~\ref{alg:rbv-consensus}. A correct
process \emph{proposes} a binary value $b$ by invoking
$\op{rbc-propose}(b)$; the consensus abstraction \emph{decides} for $b$
through an $\op{rbc-decide}(b)$ event.

The algorithm proceeds in rounds. In each round, an instance of
$\op{bv-broadcast}$ is invoked. A correct process $p_i$ executes
$\op{bv-broadcast}$ and waits for a value $b$ to be \op{bv-delivered},
identified by a tag characterizing the current round. When such a bit~$b$
is received, $p_i$ adds $b$ to $\var{values}$ and broadcasts $b$ through an
\str{aux} message to all processes. Whenever a process receives an
\str{aux} message containing~$b$ from $p_j$, it stores $b$ in a local set
$\var{aux}[j]$. Once $p_i$ has received a set $B \subseteq \var{values}$ of
values such that every $b \in B$ has been delivered in \str{aux} messages
from at least $n-f$ processes, then $p_i$ releases the coin for the round.
Subsequently, the process waits for the coin protocol to output a binary
value~$s$ through $\op{output-coin}(s)$, tagged with the current round
number.

Process $p_i$ then checks if there is a single value $b$ in $B$.  If so,
and if $b=s$, then it decides for value~$b$.  The process then proceeds to
the next round with proposal $b$.  If there is more than one value in $B$,
then $p_i$ changes its proposal to $s$.  In any case, the process starts
another round and invokes a new instance of $\op{bv-broadcast}$ with its
proposal.  Note that the protocol appears to execute rounds forever.

\subsection{A liveness problem}

Tholoniat and Gramoli~\cite{TG19} mention a liveness issue with the
randomized algorithm in the PODC-14
version~\cite{DBLP:conf/podc/MostefaouiMR14}, as presented in the previous
section. They sketch a problem that may prevent progress by the correct
processes when the messages between them are received in a specific order.
In the JACM-15 version, Most{\'{e}}faoui \emph{et
  al.}~\cite{DBLP:journals/jacm/MostefaouiMR15} appear to be aware of the
issue and present a different, more complex consensus protocol.

We give a detailed description of the problem in
Algorithm~\ref{alg:rbv-consensus}.  Recall the implementation of
binary-value broadcast, which disseminates bits in \str{value} messages.
According to our model, the processes communicate by exchanging messages
through an asynchronous reliable point-to-point network.  Messages may be
reordered, as in the PODC-14 version.

Let us consider a system with $n = 4$ processes and $f = 1$ Byzantine
process.  Let $p_1, p_2$ and $p_3$ be correct processes with input values
$0, 1, 1$, respectively, and let $p_4$ be a Byzantine process with control
over the network. Process $p_4$ aims to cause $p_1$ and $p_3$ to release
the coin with $B = \{0,1\}$, so that they subsequently propose the coin
value for the next round. If messages are scheduled depending on knowledge
of the round's coin value~$s$, it is possible, then, that $p_2$ releases
the coin with $B = \{\overline{s}\}$. Subsequently, $p_2$ proposes also
$\overline{s}$ for the next round, and this may continue forever. We now
work out the details, as illustrated in
Figures~\ref{fig:attack-1}--\ref{fig:attack-2}.

First, $p_4$ may cause $p_1$ to receive $2f+1$ $[\str{value},1]$ messages,
from $p_2, p_3$ and $p_4$, and to $\op{bv-deliver}$~$1$ sent at the start
of the
round.
Then, $p_4$ sends $[\str{value}, 0]$ to $p_3$, so that $p_3$ receives value
$0$ twice (from $p_1$ and $p_4$) and also broadcasts a $[\str{value},0]$
message itself.  Process $p_4$ also sends $0$ to $p_1$, hence, $p_1$
receives $0$ from $p_3$, $p_4$, and itself and therefore $\op{bv-delivers}$
$0$.  Furthermore, $p_4$ causes $p_3$ to $\op{bv-deliver}$ $0$ by making it
receive $[\str{value},0]$ messages from $p_1$, $p_4$, and itself.  Hence,
$p_3$ $\op{bv-delivers}$ $0$. Finally, process $p_3$ receives three
$[\str{value}, 1]$ messages (from itself, $p_2$, and $p_4$) and
$\op{bv-delivers}$ also $1$.

Recall that a process may broadcast more than one $\str{aux}$ message. In
particular, it broadcasts an $\str{aux}$ message containing a bit~$b$
whenever it has bv-delivered~$b$.  Thus, $p_1$ broadcasts first
$[\str{aux}, 1]$ and subsequently $[\str{aux}, 0]$, whereas $p_3$ first
broadcasts $[\str{aux}, 0]$ and then $[\str{aux}, 1]$.  Process $p_4$ then
sends to $p_1$ and $p_3$ $\str{aux}$ messages containing $1$ and $0$. After
delivering all six \str{aux} messages, both $p_1$ and $p_3$ finally obtain
$B = \{0,1\}$ in Line~\ref{algrbc:uponb}
(Algorithm~\ref{alg:rbv-consensus}) and see that $|B| \neq 1$ in
L~\ref{algrbc:ifb} (Algorithm~\ref{alg:rbv-consensus}). Processes $p_1$,
$p_3$ and $p_4$ invoke the common coin.

\begin{figure}
\begin{center}
\includegraphics[height=6cm]{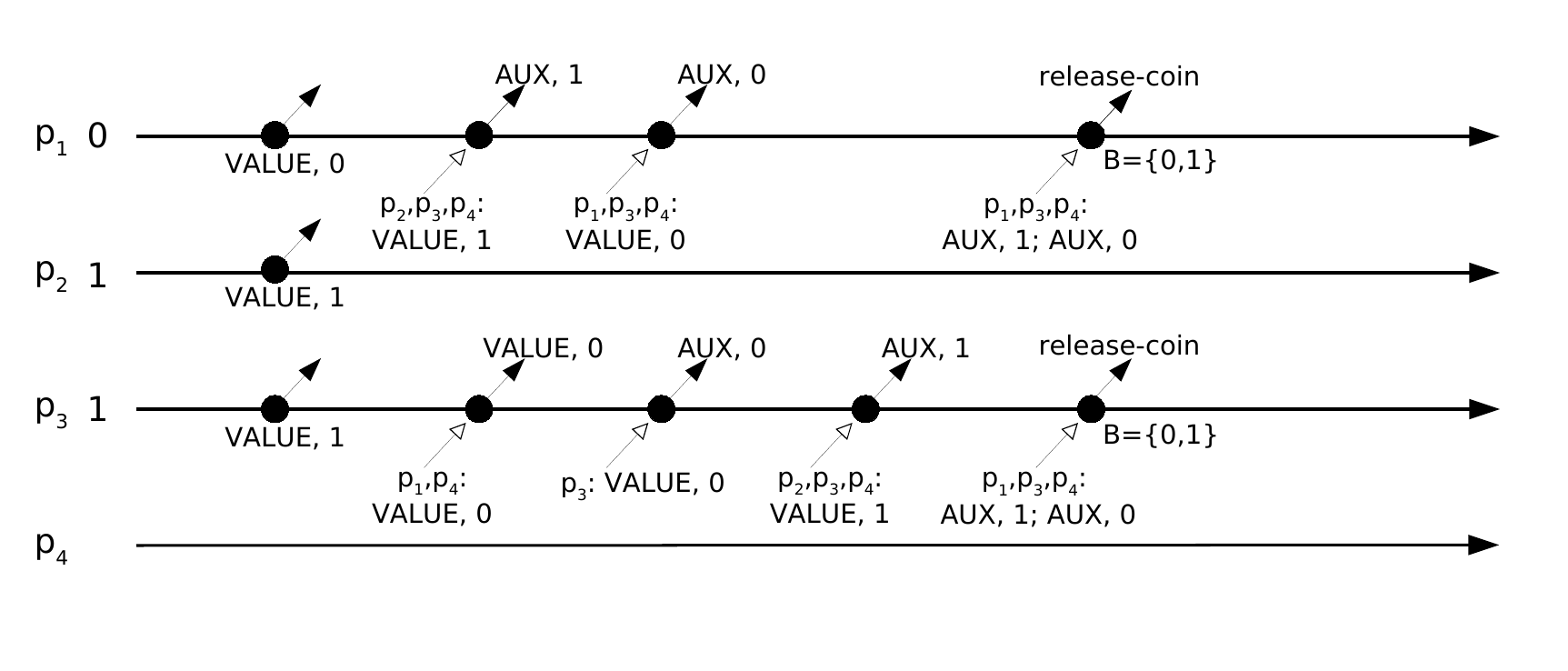}
\end{center}
\caption{The execution of Algorithm~\ref{alg:rbv-consensus}, where processes $p_1$ and $p_3$ execute Line~\ref{algrbc:uponb} with $B=\{0,1\}$.}
\label{fig:attack-1}
\end{figure}

The Byzantine process~$p_4$ may learn the coin value as soon as $p_1$ or
$p_3$ have released the common coin, according to unpredictability.  Let
$s$ be the coin output. We distinguish two cases:

\begin{description}
\item[Case $s = 0$:] Process $p_2$ receives now three $[\str{value}, 1]$
  messages, from $p_3$, $p_4$ and itself, as shown in
  Figure~\ref{fig:attack-2}.  It $\op{bv-delivers}$~$1$ and broadcasts an
  $[\str{aux}, 1]$ message. Subsequently, $p_2$ delivers three $\str{aux}$
  messages containing $1$, from $p_1, p_4$ and itself, but no
  $[\str{aux}, 0]$ message. It follows that $p_2$ obtains $B = \{1\}$ and
  proposes 1 for the next round in Line~\ref{algrbc:bvbit}.  On the other
  hand, $p_1$ and $p_3$ adopt $0$ as their new proposal for the next round,
  according to Line~\ref{algrbc:bvcoin}.  This means that no progress was
  made within this round. The three correct processes start the next round
  again with differing values, again two of them propose one bit and the
  remaining one proposes the opposite.

\begin{figure}
\begin{center}
\includegraphics[height=6cm]{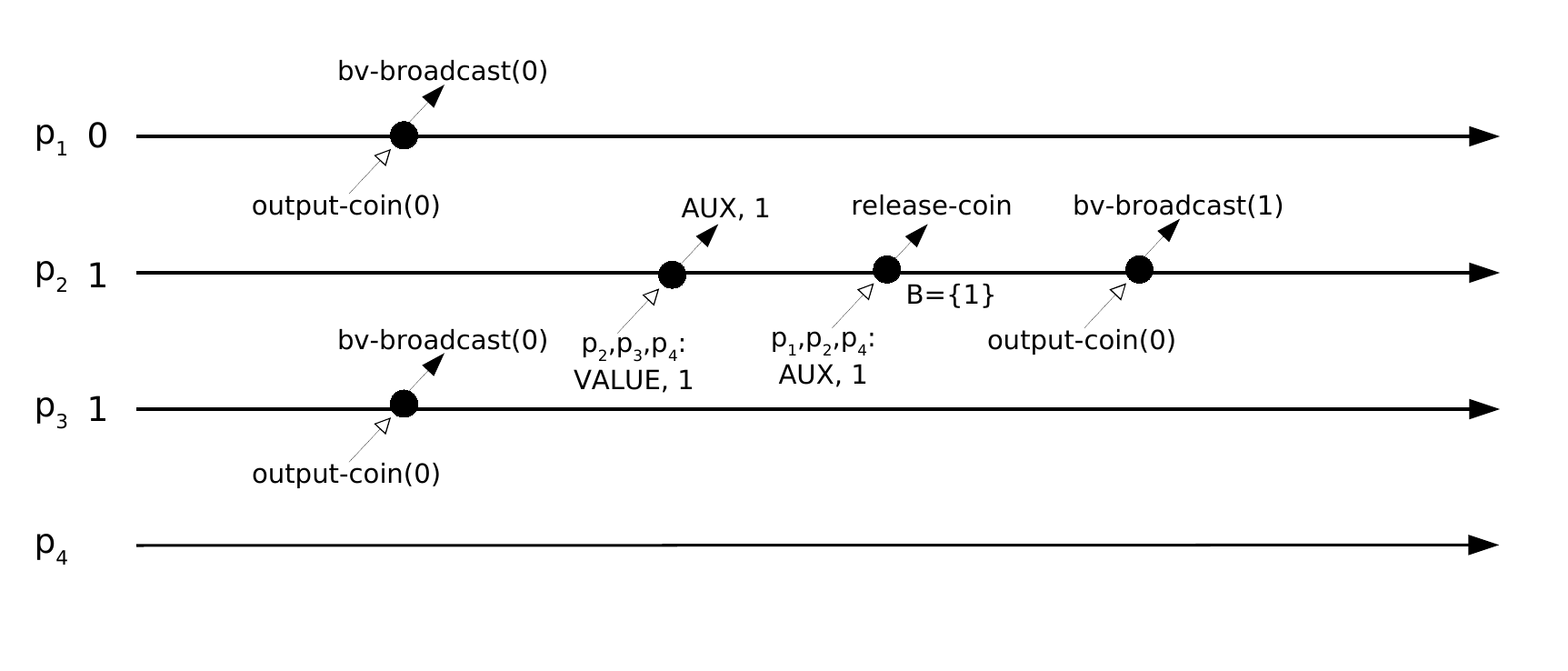}
\end{center}
\caption{Continuing the execution for the case $s=0$: Process~$p_2$
  executes Line~\ref{algrbc:uponb} with $B=\{1\}$.  Processes $p_1$ and
  $p_3$ have already proposed the coin value $s=0$ for the next round, but
  $p_2$ proposes $\overline{s} = 1$.}
\label{fig:attack-2}
\end{figure}

\item[Case $s = 1$:] Process $p_4$ sends $[\str{value}, 0]$ to $p_2$, so
  that it delivers two \str{value} messages containing $0$ (from $p_1$ and
  $p_4$) and thus also broadcast $[\str{value}, 0]$ (this execution is not
  shown).  Recall that $p_3$ has already sent $[\str{value}, 0]$ before.
  Thus, $p_2$ receives $n-f$ $[\str{value}, 0]$ messages,
  $\op{bv-delivers}$ $0$, and also broadcasts an $\str{aux}$ message
  containing $0$. Subsequently, $p_2$ may receive $n-f$ messages
  $[\str{aux}, 0]$, from $p_3$, $p_4$, and itself. It follows that $p_2$
  executes Line~\ref{algrbc:uponb} with $B = \{0\}$ and chooses~0 as its
  proposal for the next round (in Line~\ref{algrbc:bvbit}). On the other
  hand, also here, $p_1$ and $p_3$ adopt the coin value~$s=1$ and propose~1
  for the next round in Line~\ref{algrbc:bvcoin}. Hence, no progress has
  been made in this round, as the three correct processes enter the next
  round with differing values.
\end{description}

\noindent
The protocol may continue like this forever, producing an infinite
execution with no termination.

\subsection{Fixing the problem}

We show how the problem can be prevented with a conceptual insight and two
small changes to the original protocol. We do this by recalling the example
just presented. The complete protocol and a formal proof are given in
Section~\ref{sec:consensus}, using the more general model of asymmetric
quorums.

We start by considering the nature of the common coin abstraction: In any
full implementation, the coin is not an abstract oracle, but implemented by
a concrete protocol that exchanges messages among the processes.

Observe now that in the problematic execution, the network reorders
messages between correct processes. Our first change, therefore, is to
assume FIFO ordering on the reliable point-to-point links. This may be
implemented over authenticated links, by adding sequence numbers to
messages and maintaining a buffer at the
receiver~\cite{DBLP:books/daglib/0025983}.  Consider $p_2$ in the example
and the messages it receives from the other correct processes, $p_1$ and
$p_3$. W.l.o.g.~any protocol implementing a common coin requires an
additional message exchange, where a correct process sends at least one
message to every other process, say, a \str{coin} message with arbitrary
content (to be specific, see Algorithm~\ref{alg:acc},
Section~\ref{sec:consensus}). Observe that at least $f+1$ \emph{correct}
processes are required to send a \str{coin} message in order for a process
to have information about the coin value.

When $p_2$ waits for the output of the coin, it needs to receive, again
w.l.o.g., a \str{coin} message from $n-f$ processes.  Since the other two
correct processes ($p_1$ and $p_3$) have sent two \str{value} messages and
\str{aux} messages each before releasing the coin, then $p_2$ receives
these messages from at least one of them before receiving enough \str{coin}
messages, according to the overlap among Byzantine quorums.

This means that $p_2$ cannot satisfy the condition in
Line~\ref{algrbc:uponb} with $|B| = 1$.  Thus the adversary may no longer
exploit its knowledge of the coin value to prevent
termination. (Most{\'{e}}faoui \emph{et
  al.}~\cite{DBLP:journals/jacm/MostefaouiMR15} (JACM-15) remark in
retrospect about the PODC-14 version that a ``fair scheduler'' is needed.
However, this comes without any proof and thus remains open, especially
because the JACM-15 version introduces a much more complex version of the
protocol.)

Our second change is to allow the set $B$ to dynamically change while the
coin protocol executes. In this way, process $p_2$ may find a suitable $B$
according to the received \str{aux} messages while concurrently running the
coin protocol.  Eventually, $p_2$ will have output the coin \emph{and} its
set $B$ will contain the same values as the sets $B$ of $p_1$ and
$p_3$. Observe that this dynamicity is necessary; process $p_2$ could start
to release the coin after receiving $n-f$ \str{aux} messages containing
only the value $1$. However, following our example, due to the assumed FIFO
order, it will receive from another correct process also an $\str{aux}$
message containing the value $0$, before the \str{coin} message. If we do
not ask for the dynamicity of the set $B$, process $p_2$, after outputting
the coin, will still have $|B|=1$. Most{\'{e}}faoui \emph{et al.} in the
PODC-14 version (\cite[Fig.~2, Line~5]{DBLP:conf/podc/MostefaouiMR14})
seem to rule this out.

Observe that the common-coin primitive here requires \emph{more than $f$
  correct} processes to release the coin before it may be predicted by the
faulty processes.  Within an implementation, this translates into receiving
a \str{coin} message from more than $2f$ processes (or $2f+1 = n-f$
processes, in case $n = 3f+1$).  Abraham \emph{et
  al.}~\cite{DBLP:conf/podc/AbrahamBY22,cryptoeprint:2022/711} show that
such an assumption (which they call an $2f$-unpredictable coin) is
necessary in order to prevent this liveness problem.  With an ordinary coin
primitive (i.e., one where at least \emph{one correct process} is required
to send a \str{coin} message, before information about the coin value may
become available), an adversary would still be able to produce an infinite
execution and to violate termination~\cite[Appendix
A]{cryptoeprint:2022/711}.

\section*{Acknowledgments}

The authors are grateful to Joachim Neu and Srivatsan Sridhar for pointing
out a mistake in an earlier version of the paper as well as for interesting
discussions about asymmetric and flexible trust.

This work has received funding from the Swiss National Science Foundation
(SNSF) under grant agreement Nr\@.~200021\_188443 (Advanced Consensus
Protocols).

\bibliography{references, dblpbibtex}
\bibliographystyle{ieeesort}

\end{document}